\documentclass[12pt]{iopart}

\usepackage{iopams}

\usepackage[utf8]{inputenc}
\usepackage[T1]{fontenc}
\usepackage[table,usenames,dvipsnames]{xcolor}
\usepackage{bbm}
\usepackage{amsthm}
\usepackage{graphicx}
\usepackage[]{caption}
\usepackage{multirow}
\usepackage[doi=false,isbn=false,url=false,eprint=false,backend=biber]{biblatex}

\usepackage{nicefrac}

\usepackage[ruled, linesnumbered]{algorithm2e}
\SetKwComment{Comment}{/* }{ */}
\SetKwInOut{Input}{input}
\SetKwInOut{Output}{output}
\DontPrintSemicolon

\usepackage[hidelinks]{hyperref}

\newtheorem{proposition}{Proposition}%
\newtheorem{assumption}{Assumption}%
\newtheorem{example}{Example}%
\newtheorem{definition}{Definition}%

\newcommand{\diff}[1]{\mathrm{d}#1}

\newcommand{\diag}[1]{\mathrm{diag}\!\left(#1\right)}
\newcommand{\Hell}[2]{H\!\left( {#1}, {#2} \right) }
\newcommand*{\argmin}{\mathrm{arg\,min}}
\newcommand{\norm}[1]{\left\Vert#1\right\Vert}

\newcommand{\sigobs}{\Sigma_{\mathrm{obs}}} 
\newcommand{\sigpr}{\Sigma_0} 
\newcommand{\mupr}{\mu_0} 
\newcommand{\sigpo}{\Sigma} 
\newcommand{\mupo}{\mu} 

\newcommand{\CCSmix}{\tilde\pi_r^{\mathrm{CCS}}(w|y)}
\newcommand{\CCSmixmarg}{\tilde\pi_r^{\mathrm{CCS}}(w_\mathcal{I}|y)}
\newcommand{\MAPmix}{\tilde\pi_r^{\mathrm{MAP}}(w|y)}
\newcommand{\MAPmixmarg}{\tilde\pi_r^{\mathrm{MAP}}(w_\mathcal{I}|y)}

\newcommand{\CCSposX}{\tilde\pi^{\mathrm{CCS(X)}}_r(x|y)}
\newcommand{\CCSmixmargX}{\tilde\pi_r^{\mathrm{CCS}}(x_\mathcal{I}|y)}
\newcommand{\CCSposW}{\tilde\pi^{\mathrm{CCS(W)}}_r(x|y)}
\newcommand{\MAPposX}{\tilde\pi^{\mathrm{MAP(X)}}(x|y)}
\newcommand{\MAPposW}{\tilde\pi^{\mathrm{MAP(W)}}_r(x|y)}

\definecolor{DTUred}{RGB}{153, 0, 0}
\definecolor{green}{RGB}{32, 110, 22}
\begin{document}
	\title[
	Continuous Gaussian mixture solution for linear Bayesian inversion
	]{
	Continuous Gaussian mixture solution
	for linear Bayesian inversion with application to Laplace priors
}


\author{R Flock$^1$, Y Dong$^1$, F Uribe$^2$ and O Zahm$^3$}

\address{$^1$ Department of Applied Mathematics and Computer Science, Technical University of Denmark, Richard Petersens Plads, Building 324, Kongens Lyngby, 2800, Denmark}
\address{$^2$ School of Engineering Sciences, Lappeenranta-Lahti University of Technology, Yliopistonkatu 34, Lappeenranta, 53850, Finland}
\address{$^3$ Univ. Grenoble Alpes, Inria, CNRS, Grenoble INP, LJK, 38000, Grenoble, France}

\ead{raff@dtu.dk}

\vspace{10pt}
\begin{indented}
	\item[]September 2024
\end{indented}

\begin{abstract}
	We focus on Bayesian inverse problems with Gaussian likelihood, linear forward model, and priors that can be formulated as a Gaussian mixture. 
	Such a mixture is expressed as an integral of Gaussian density functions weighted by a mixing density over the mixing variables. 
	Within this framework, the corresponding posterior distribution also takes the form of a Gaussian mixture, and we derive the closed-form expression for its posterior mixing density. 
	To sample from the posterior Gaussian mixture, we propose a two-step sampling method. 
	First, we sample the mixture variables from the posterior mixing density, and then we sample the variables of interest from Gaussian densities conditioned on the sampled mixing variables. 
	However, the posterior mixing density is relatively difficult to sample from, especially in high dimensions. 
	Therefore, we propose to replace the posterior mixing density by a dimension-reduced approximation, and we provide a bound in the Hellinger distance for the resulting approximate posterior. 
	We apply the proposed approach to a posterior with Laplace prior, where we introduce two dimension-reduced approximations for the posterior mixing density. 
	Our numerical experiments indicate that samples generated via the proposed approximations have very low correlation and are close to the exact posterior.

\end{abstract}

%
\vspace{2pc}
\section{Introduction}
%
In Bayesian inverse problems, the posterior density on the unknown parameter $x$, given the data $y$, is formulated using Bayes' theorem as the product of the likelihood function $x\mapsto\pi(y|x)$ and the prior density $\pi(x)$ as
$$
\pi(x|y) \propto \pi(y|x) \pi(x).
$$
The likelihood function incorporates the forward model, the observed data and the assumed noise model. In this paper, we focus on likelihood functions with linear forward model and additive Gaussian noise (linear-Gaussian likelihood). In this setting, it is well-known that if the prior is Gaussian, the posterior is also Gaussian, e.g., \cite[p.78]{kaipioStatisticalComputationalInverse2005}. For other priors, oftentimes no closed-form expressions for the posterior exist. 

We express our non-Gaussian prior as a \emph{mixture model} \cite{lovricInternationalEncyclopediaStatistical2011} of the form
$$
\pi(x) = \int \pi(x|w) \pi(w) \diff{w},
$$
where $w$ is a \emph{mixing variable} with density $\pi(w)$. Gaussian mixture models correspond to the case where $\pi(x|w)$ is Gaussian, and they encompass the \emph{finite Gaussian mixture model} (known as GMM, e.g., \cite{gelmanBayesianDataAnalysis2013}) and \emph{Gaussian scale mixtures}, e.g., \cite{andrewsScaleMixturesNormal1974, westScaleMixturesNormal1987, gneitingNormalScaleMixtures1997, chuEstimationDecisionLinear1972}. 
In the Bayesian context, GMMs are a valuable tool to model multimodal priors and are successfully applied in practice, for example in geophysics to model properties of the subsurface \cite{grana*BayesianGaussianMixture2015} and in patch-based image denoising to model a prior on image patches \cite{houdardHighDimensionalMixtureModels2018}. Gaussian scale mixtures are used in Bayesian regression analysis to model hierarchical heavy-tailed error models and priors \cite{westOutlierModelsPrior1984, baeGeneSelectionUsing2004}. In particular, using a Gaussian scale mixture representation for the Laplace prior leads to the well-known Bayesian LASSO \cite{parkBayesianLasso2008d}. Hierarchical models based on Gaussian scale mixtures are also frequently used in connection with image reconstruction \cite{calvettiGaussianHypermodelRecover2007, calvettiHypermodelsBayesianImaging2008}. More recently, \cite{uribeHorseshoePriorsEdgePreserving2023} employed a Gaussian scale mixture to propose the extended horseshoe prior model, and \cite{senchukovaBayesianInversionStudent2024} used a Gaussian scale mixture representation of the Student's t-distribution to construct a hierarchical Markov random field prior. 

In the first part of this paper, we derive a general closed-form expression of the posterior as a continuous Gaussian mixture of the form
$$
\pi(x|y) = \int \pi(x|w,y) \pi(w|y) \diff{w},
$$
under the assumptions that the prior $\pi(x)$ is a Gaussian mixture, and the likelihood $\pi(y|x)$ is linear-Gaussian.
The Gaussian posterior mixture naturally gives rise to a simple algorithm to obtain samples from the posterior, which consists of two sampling steps. In the first step, one samples the mixing variable $w$ from the posterior mixing density $\pi(w|y)$, and in the second step one samples the Gaussian $\pi(x|w,y)$ conditioned on the previously sampled mixing variable. We propose to replace the posterior mixing density by an easy-to-sample and dimension-reduced approximation $\tilde\pi_r(w|y)$ and show that the Hellinger distance between $\pi(x|y)$ and the approximation
$$
\tilde\pi_r(x|y) = \int \pi(x|w,y) \tilde\pi_r(w|y) \diff{w}
$$ 
is bounded by the Hellinger distance between $\pi(w|y)$ and $\tilde\pi_r(w|y)$.

In the second part of this paper, we focus on applying our framework to the case of Laplace prior, which can be formulated as a Gaussian scale mixture. Besides its usage in Bayesian regression analysis, the Laplace prior can be employed in Bayesian inference for signal processing. This is because natural signals can be effectively represented in a sparse manner using adapted bases, such as point source bases and wavelets bases \cite{caiUncertaintyQuantificationRadio2018b}. Then, one can use the so-called \emph{synthesis formulation} $s=W x$ to expand a signal $s$ in a suitable basis $W$ for which $x$ is sparse \cite{eladAnalysisSynthesisSignal2007}. In this case, the heavy-tailed Laplace distribution is a typical prior choice to enforce sparsity in $x$.

For the case of Laplace prior, we formulate two dimension-reduced approximations $\tilde\pi_r(w|y)$.
The first one is obtained by employing the certified coordinate selection (CCS) method \cite{flockCertifiedCoordinateSelection2024}, and the second one is based on a maximum a posteriori probability (MAP) estimate similar to the Laplace approximation, see, e.g., \cite{murphyMachineLearningProbabilistic2012}.
We illustrate with two numerical experiments that both approximations lead to computationally efficient and nearly uncorrelated samples from the approximate posterior $\tilde\pi_r(x|y)$.
Moreover, the estimates are close to our reference solutions.

\paragraph{\texorpdfstring{\bf Contributions.}{Contributions}}
Our main contributions are as follows.
\begin{itemize}
	\item We derive a closed-form expression for continuous Gaussian posterior mixtures resulting from posteriors with a linear-Gaussian likelihood and a Gaussian prior mixture.
	\item 
	Leveraging the Gaussian posterior mixture formulation, we describe two easy-to-sample approximations to a posterior with Laplace prior, where we replace the posterior mixing density $\pi(w|y)$ with a dimension-reduced approximation $\tilde\pi_r(w|y)$.
	\item We numerically illustrate the benefits of the Gaussian mixture formulation by applying our dimension-reduced posterior approximations to a 1D deblurring example and a 2D high-dimensional super-resolution example with Laplace priors.
\end{itemize}

\paragraph{\texorpdfstring{\bf Outline.}{Outline}}
In Section \ref{sec:gen_form}, we derive the closed-form expression of the Gaussian posterior mixture for posteriors with a linear-Gaussian likelihood and a Gaussian prior mixture. In Section \ref{sec:app_lap}, we focus on the Laplace prior and the corresponding Gaussian posterior mixture. Here we outline two dimension-reduced approximations $\tilde\pi_r(w|y)$ for the posterior mixing density $\pi(w|y)$. In Section \ref{sec:num_exp}, we test our approximations from the previous section on two synthetic examples: A 1D deblurring example and a 2D high-dimensional super-resolution example. We draw our conclusions in Section \ref{sec:conc}.
\section{General framework}\label{sec:gen_form}
In this section, under a certain assumption, we first introduce the general posterior mixture, which can be formulated given any likelihood function and prior mixture. Subsequently, we derive the closed-form expression of the Gaussian posterior mixture resulting from a linear-Gaussian likelihood and a Gaussian prior mixture. We end the section with a simple algorithm for sampling the Gaussian posterior mixture.
\subsection{Posterior mixture}
Bayesian inverse problems are formulated via the posterior probability density
\begin{equation}\label{eq:gen_pos}
	\pi(x|y) = \frac{\pi(y|x) \pi(x)}{\pi(y)},
\end{equation}
where $x\in\mathbb{R}^d$ is the vector of unknown parameters and $y\in\mathbb{R}^m$ is the observed data. As discussed in the introduction, $x\mapsto \pi(y|x)$ is the likelihood function of observing the data $y$, $\pi(x)$ is the prior density, and $\pi(y)$ is the data density, which for some fixed $y$ is sometimes called the \emph{(model) evidence}.
\begin{definition}[Mixture \cite{lovricInternationalEncyclopediaStatistical2011}]\label{def:mix}
	A density $\pi(x)$ defined as
	\begin{eqnarray}\label{eq:mix_model}
		\pi(x) = \int \pi(x|w) \pi(w) \diff{w},
	\end{eqnarray}
	is a \emph{mixture} with \emph{component density} $\pi(x|w)$, \emph{mixing density} $\pi(w)$ and \emph{mixing variable} $w$.
\end{definition}

Before we state our general posterior mixture, we need the assumption that $w$ and $y$ are conditionally independent given $x$.
\begin{assumption}\label{ass:gen_pos_mix_ass}
	The data $y$ and the mixing variable $w$ are conditionally independent given the parameter $x$, meaning
	\begin{equation}\label{eq:gen_pos_mix_ass}
		\pi(y,w|x) = \pi(y|x)\pi(w|x).
	\end{equation}
\end{assumption}
The random variable $w$ can be considered to be a hyperparameter, and thus only a property of the prior density on $x$. 
Therefore, it is natural to assume that given a realization of $x$, the data $y$ and $w$ should be independent.
Under Assumption \ref{ass:gen_pos_mix_ass}, the joint density of $(x,y,w)$ is properly defined as the product $\pi(x,y,w)  = \pi(y|x)\pi(x|w)\pi(w)$.
In this case, we can represent the posterior $\pi(x|y)$ as a mixture as follows.
\begin{proposition}[Posterior mixture]\label{prop:pos_mix}
	Under Assumption \ref{ass:gen_pos_mix_ass}, we can formulate the posterior mixture
	\begin{equation}\label{eq:gen_pos_mix}
		\pi(x|y) = \int \pi(x|w,y) \pi(w|y) \diff{w},
	\end{equation}
	where
	\begin{eqnarray}
		\pi(x|w,y) &\propto \pi(y|x)\pi(x|w) \label{eq:gen_pos_mix_1}\\
		\pi(w|y) &\propto \pi(w)\int \pi(y|x')\pi(x'|w) \diff{x'} \label{eq:gen_pos_mix_2}
	\end{eqnarray}
	are respectively called posterior component density and posterior mixing density.
\end{proposition}
\begin{proof}
	See \ref{proof:pos_mix}.
\end{proof}
The posterior mixture formulation (\ref{eq:gen_pos_mix}) can be used to sample from $\pi(x|y)$ by first drawing $w^{(i)}\sim\pi(w|y)$ and then $x^{(i)}\sim\pi(x|w^{(i)},y)$. While sampling from $\pi(x|w^{(i)},y)$ in (\ref{eq:gen_pos_mix_1}) is straightforward when the likelihood $\pi(y|x)$ and the prior component $\pi(x|w)$ are \emph{conjugate distributions}, sampling from $\pi(w|y)$ can be difficult in general. Therefore, we consider replacing $\pi(w|y)$ with an easy-to-sample approximation $\tilde\pi(w|y)$. The following proposition shows that we can control the resulting approximate posterior $\tilde\pi(x|y)$ if we can control $\tilde\pi(w|y)$ in the Hellinger distance.
\begin{proposition}[Certified approximate posterior mixture]\label{prop:bound_pos_mix}
	Let $\tilde\pi(w|y)$ be an approximation to $\pi(w|y)$ in (\ref{eq:gen_pos_mix_2}), and let $\tilde\pi(x|y) = \int \pi(x|w,y) \tilde\pi(w|y) \diff{w}$ be the resulting approximate posterior. By controlling the Hellinger distance
	\begin{equation*}
		\Hell{\pi(w|y)}{\tilde\pi(w|y)}^2 = \frac{1}{2} \int \left( \sqrt{ \pi(w|y) }-\sqrt{ \tilde\pi(w|y) } \right)^2 \diff{w} \leq \epsilon,
	\end{equation*}
	for a given $\epsilon>0$, we can ensure 	\begin{equation}\label{eq:Hell_pos_mix}
		\Hell{\pi(x|y)}{\tilde\pi(x|y)}^2 \leq \epsilon.
	\end{equation}
\end{proposition}
\begin{proof}
	See \ref{proof:hell_bound}.
\end{proof}
\subsection{Gaussian posterior mixture}
We now consider the case where the component density $\pi(x|w)$ is Gaussian.
\begin{definition}[Gaussian mixture]\label{def:gauss_mix}
	A Gaussian mixture $\pi(x)$ is a mixture as in Definition \ref{def:mix} with Gaussian component density $\pi(x|w)$, meaning
	\begin{equation}\label{eq:gauss_mix_model}
		\pi(x|w) \propto \exp\left(-\frac12 \| x-\mupr(w) \|_{\sigpr(w)^{-1}}^2 \right)  .
	\end{equation}
	Here $\mupr(w)$ is the component mean vector and $\sigpr(w)$ is the component covariance matrix. Specifically, if $\mupr(w)=0$, the mixture is called a Gaussian scale mixture \cite{andrewsScaleMixturesNormal1974, gneitingNormalScaleMixtures1997}.
\end{definition}
\begin{example}[Gaussian mixture model (GMM)]\label{ex:GMM}
	In the literature, the term Gaussian mixture model (GMM) is commonly used for Gaussian mixtures with a finite number of Gaussian component densities, e.g., \cite{murphyMachineLearningProbabilistic2012}.
	In this case, the mixing density $\pi(w) = \sum_{i=1}^N \delta_{i}(w) p_i$ reduces to a sum of Dirac masses with weights $p_i\geq0$, $\sum_{i=1}^N p_i=1$, such that
	\begin{equation*}
		\pi(x) = \sum_{i=1}^N \frac{1}{\sqrt{2\pi\det\Sigma_i}} \exp\left(-\frac12 \| x-\mu_i \|_{\Sigma_i^{-1}}^2 \right) p_i.
	\end{equation*}
\end{example}
\begin{example}[Laplace distribution]\label{ex:Laplace}
	The univariate probability density function of the Laplace distribution $\pi(x)  = \frac{\delta}{2} \exp\left( -\delta |x| \right)$ with zero mean and rate parameter $\delta>0$ can be written as Gaussian scale mixture \cite{andrewsScaleMixturesNormal1974}
	\begin{equation}\label{eq:un_lap_mix}
		\pi(x)  = \int_{\mathbb{R}_{>0}} \frac{1}{\sqrt{2\pi w}} \exp\left( -\frac{1}{2} \frac{x^2}{w} \right) \pi(w) \diff{w} ,
	\end{equation}
	with the exponential mixing density $\pi(w)=\lambda \exp( -\lambda w )$ with rate parameter $\lambda=\delta^2/2$.
\end{example}
\begin{example}[Student's t-distribution]\label{ex:Student}
	The univariate probability density function of the Student's t-distribution $\pi(x) \propto (1 + \frac{t^2}{\nu} )^{-(\nu+1)/2} $ with $\nu$ degrees of freedom can also be written as Gaussian scale mixture (\ref{eq:un_lap_mix}).
	It has an inverse-gamma mixing density $\pi(w)\propto w^{-\alpha-1} \exp( -\frac{\beta}{w} )$ with shape parameter $\alpha=\nu/2$ and rate parameter $\beta=\nu/2$ \cite{andrewsScaleMixturesNormal1974}.
\end{example}
Now let the likelihood $\pi(y|x)$ be linear-Gaussian and given by
\begin{equation}\label{eq:lin_gauss_like}
	\pi(y|x) \propto \exp\left( -\frac12 \|Ax-y\|_{\sigobs^{-1}}^2 \right) ,
\end{equation}
with $A\in\mathbb{R}^{m\times d}$ and a positive definite covariance matrix $\sigobs\in\mathbb{R}^{m\times m}$. 
This likelihood corresponds to the data-generating model $y=Ax+\varepsilon$, where $\varepsilon\sim\mathcal{N}(0,\sigobs)$ is Gaussian noise. 
A posterior with the linear-Gaussian likelihood (\ref{eq:lin_gauss_like}) and a prior formulated as Gaussian mixture (\ref{eq:gauss_mix_model}) can be formulated as Gaussian posterior mixture. We give the explicit expressions for the posterior component density and the posterior mixing density in the following proposition.
\begin{proposition}[Gaussian posterior mixture]\label{prop:Gauss_pos_mix}
	Assume a linear-Gaussian likelihood $\pi(y|x)$ as in (\ref{eq:lin_gauss_like}) and a Gaussian prior mixture $\pi(x)$ as in Definition \ref{def:gauss_mix}. Under Assumption \ref{ass:gen_pos_mix_ass}, the posterior $\pi(x|y)$ is a Gaussian mixture $\pi(x|y) = \int \pi(x|w,y) \pi(w|y)\diff{w}$ with Gaussian component density 
	\begin{equation}\label{eq:pos_comp}
		\pi(x|w,y) \propto \exp\left(-\frac{1}{2}\|x-\mu(w,y)\|_{\Sigma(w)^{-1}}^2\right),
	\end{equation}
	where
	\begin{eqnarray}
		\sigpo(w)^{-1} &= A^\mathsf{T}\sigobs^{-1}A + \sigpr(w)^{-1} \label{eq:conpo_sig} \\
		\mupo(w,y) &= \sigpo(w)\left( A^\mathsf{T} \sigobs^{-1}y + \sigpr(w)^{-1} \mupr(w) \right), \label{eq:conpo_mean}
	\end{eqnarray}
	and with the mixing density
	\begin{equation}\label{eq:pos_mixing}
		\pi(w|y) \propto \pi(y|w) \pi(w),
	\end{equation}
	where
	\begin{equation}\label{eq:pos_mixing_like}
		\pi(y|w) \propto \sqrt{\frac{\det \sigpo(w)}{\det \sigpr(w)}} \exp\left(\frac12 \| \mupo(w,y) \|_{\sigpo(w)^{-1}}^2 - \frac12 \| \mupr(w) \|_{\sigpr(w)^{-1}}^2 \right).
	\end{equation}
\end{proposition}
\begin{proof}
	See \ref{proof:Gauss_pos_mix}.
\end{proof}

We notice that $\pi(y|w)$ in (\ref{eq:pos_mixing_like}) does not depend on the mixing density 
$\pi(w)$. 
This means that we can easily adapt the posterior mixture 
to any choice of Gaussian prior mixture by replacing $\pi(w)$ in (\ref{eq:pos_mixing}), e.g., using the Gaussian mixtures in Examples \ref{ex:Laplace} and \ref{ex:Student}.
\subsection{General sampling algorithm}
We can draw samples $\{x^{(i)}\}_{i\geq1}$ from the posterior mixture (\ref{eq:gen_pos_mix}) as shown in Algorithm \ref{alg:gen}. 
This is because Algorithm \ref{alg:gen} yields samples $\{(w^{(i)},x^{(i)})\}_{i\geq1}$ from the joint probability density $\pi(w,x|y) = \pi(x|w,y)\pi(w|y)$, whose marginal with respect to $x$ is precisely $\pi(x|y)$.

\begin{algorithm}
	\For{$i=1$ \KwTo $N$}{
		Sample $w^{(i)}$ from the posterior mixing density $\pi(w|y)$ \label{algline:gen_draw_hard_dens}.\;
		Sample $x^{(i)}$ from the Gaussian posterior component density $\pi(x|w^{(i)},y)$\label{algline:gen_draw_eas_dens}.\;
	}
	Return samples $\{x^{(i)}\}_{i=1}^N$.
	\caption{General algorithm to sample the posterior mixture density (\ref{eq:gen_pos_mix})}
	\label{alg:gen}
\end{algorithm}

Sampling from the posterior component density $\pi(x|w,y)$ in the second step is straightforward since it is Gaussian, see Proposition \ref{prop:Gauss_pos_mix}. Moreover, once samples $w^{(i)}\sim\pi(w|y)$ are available, the second sampling step can be parallelized. In fact, efficient sampling from the Gaussian posterior mixture hinges on the ability of sampling from the posterior mixing density $\pi(w|y)$ in (\ref{eq:pos_mixing}) in the first step. In the following, we discuss the two sampling steps in detail.
\subsubsection{Dimension reduction in the posterior mixing density.}\label{sec:sam_pos_mixing}
Sampling the posterior mixing density $\pi(w|y)$ can be computationally expensive, since its evaluation requires the computation of determinants and solutions to linear systems of equations. 
This is further exacerbated when the dimensions of the parameter $x$ and the mixing component $w$ are high.
For these reasons, we propose to replace $\pi(w|y)$ by an approximation $\tilde\pi_r(w|y)$ obtained via the certified dimension reduction (CDR) method introduced in \cite{zahmCertifiedDimensionReduction2022b}, see also \cite{cui2022unified,li2023principal}. 
Note that other approximations can be used, for example, a Laplace approximation, see, e.g., \cite{murphyMachineLearningProbabilistic2012}, or the prior mixing density $\pi(w)$.

Applying the CDR method to the posterior mixing density $\pi(w|y)$ in (\ref{eq:pos_mixing}) consists in replacing $\pi(y|w)$ in (\ref{eq:pos_mixing_like}) by a \emph{ridge} approximation $w\mapsto \tilde{\pi}(y|U_r^\mathsf{T} w)$.
The matrix $U_r\in\mathbb{R}^{d\times r}$ is pre-determined and has $r\ll d$ orthogonal columns, where $d$ is the dimension of the mixing variable $w$ (in our experiments, both $w$ and $x$ are of dimension $d$).
This way, the effective dimension $r$ of the approximation is much smaller than the original, which reduces the computational cost significantly. 
Consequently, the approximated posterior mixing density takes the form of
$
\tilde{\pi}(w|y) \propto \tilde{\pi}_r(y|U_r^\mathsf{T} w) \pi(w),
$
where $U_r$ is constructed by minimizing an upper-bound of the Kullback--Leibler (KL) divergence of $\pi(w|y)$ from $\tilde{\pi}(w|y)$. 
This procedure has been generalized to the broader class of $\alpha$-divergences in \cite{li2023principal}. 

However, the theoretical foundations of the CDR method requires the mixing density $\pi(w)$ to satisfy a logarithmic Sobolev inequality, which is trivially satisfied by the Gaussian $\pi(w)\propto\exp(-\|w\|^2/2)$, but not by, e.g., the exponential density $\pi(w)\propto\exp(-\lambda w)$ associated with the Laplace prior, see Example \ref{ex:Laplace}. One way to circumvent this issue is to consider a transformation $T$, which pushes forward $\pi(w)$ to a standard normal random variable $u$ and to apply the CDR method on the transformed component $u=T(w)$, see \cite{cuiPriorNormalizationCertified2022c,verdiere2023diffeomorphism}.

We note that the Hellinger distance squared is bounded by the KL divergence, and thus, Proposition \ref{prop:bound_pos_mix} can be used to bound the approximated posterior mixture $\tilde{\pi}(x|y) = \int \pi(x|w,y) \tilde{\pi}(w|y) \diff{w}$, where $\tilde{\pi}(w|y)$ is obtained through the CDR method. 
In Section \ref{sec:pos_mixCCS}, we apply a specialized version of the CDR method, the certified coordinate selection (CCS) \cite{flockCertifiedCoordinateSelection2024}, to the case where the prior is Laplace and thus can be written as a Gaussian scale mixture. 

We conclude this section with a practical result from \cite{flockCertifiedCoordinateSelection2024}. 
Regardless of the approximation method, we can estimate the Hellinger distance based on samples $w^{(i)}\sim\tilde\pi(w|y)$ as follows:
\begin{equation*}\label{eq:num_Hell}\fl
	\Hell{\pi(w|y)}{\tilde\pi(w|y)}^2 
	\leq 2 \int \left( \sqrt{ \frac{\pi_u(w|y)}{\tilde\pi_u(w|y)} } - 1 \right)^2 \tilde\pi(w|y) \diff{w} 
	\approx \frac{2}{N} \sum_{i=1}^N \left( \sqrt{ \frac{\pi_u(w^{(i)}|y)}{\tilde\pi_u(w^{(i)}|y)} } - 1 \right)^2.
\end{equation*}
Here we use a subscript ``$u$'' to denote unnormalized densities, such that $\pi_u(w|y) \propto \pi(w|y)$, and $\tilde\pi_u(w|y) \propto \tilde\pi(w|y)$. 
This way, the above bound does not require the computation of normalizing constants.
We refer to Appendix A.3 of \cite{flockCertifiedCoordinateSelection2024} for a proof of the above inequality.

\subsubsection{Sampling the posterior component density.}\label{sec:sam_po_comp}
Sampling the posterior component density $\pi(x|w^{(i)},y)$ in (\ref{eq:pos_comp}) consists in drawing samples from the Gaussian distribution $\mathcal{N}(\mupo(w^{(i)},y), \sigpo(w^{(i)}))$. 
The standard way for exact sampling is via the Cholesky decomposition of the covariance matrix. 
However, if $\sigpo(w^{(i)})$ in (\ref{eq:conpo_sig}) is high-dimensional, it is numerically expensive to perform a single Cholesky decomposition for each sample $w^{(i)}$. 
Conversely, the Cholesky decomposition of the prior component covariance matrix $\sigpr(w)$ is oftentimes relatively easy to compute, e.g., when it has diagonal structure. 
In this case, the linear randomize-then-optimize (RTO) method proposed in \cite{bardsleyMCMCBasedImageReconstruction2012, bardsleyRandomizeThenOptimizeMethodSampling2014} can be used to obtain samples by computing the Cholesky decomposition of $\sigpr(w^{(i)})$ and solving linear systems of equations with the system matrix $\sigpo(w^{(i)})^{-1}$.
Moreover, the linear system can be reformulated as a linear least squares problem, which is better conditioned.
We describe the methodology in detail in \ref{proof:lin_RTO}.

In fact, sampling from a high-dimensional Gaussian distribution is a well explored field of research. Many methods which are tailored to specific problem structures are available. 
We refer to \cite{vonoHighDimensionalGaussianSampling2022} for an overview and references therein. 
In particular, taking into account the structure of $\sigpo(w)^{-1}$, a Gibbs sampler based on exact (GEDA) or approximate data augmentation (Approx. DA) can be used as an alternative to linear RTO.
\section{Application to Laplace prior}\label{sec:app_lap}
In this section, we employ Algorithm \ref{alg:gen} to sample from a posterior density defined by a linear-Gaussian likelihood and a Laplace prior. We obtain the Gaussian posterior mixture immediately by applying Proposition \ref{prop:Gauss_pos_mix}. 
Moreover, by employing the CCS method \cite{flockCertifiedCoordinateSelection2024}, we formulate two easy-to-sample and dimension-reduced approximate posterior mixing densities $\tilde\pi_r(w|y)$, which can be used to replace the relatively complicated density $\pi(w|y)$ in step \ref{algline:gen_draw_hard_dens} of Algorithm \ref{alg:gen}. 
\subsection{Problem formulation}
We aim at sampling from a posterior distribution defined by a linear-Gaussian likelihood (\ref{eq:lin_gauss_like}) and a product-form Laplace prior:
\begin{equation}\label{eq:lap}
	\pi(x) \propto \exp\left( -\sum_{i=1}^d \delta_i |x_i| \right),
\end{equation}
where $\delta_i>0$ for $i=1, \cdots, d$ are the rate parameters. The posterior reads
\begin{equation}\label{eq:pos_lap_pr}
	\pi(x|y) \propto \exp\left( -\frac12 \|Ax-y\|_{\sigobs^{-1}}^2 \right) \exp\left( -\sum_{i=1}^d \delta_i |x_i| \right).
\end{equation}

Recall from Example \ref{ex:Laplace} that the mixing density of a univariate Laplace distribution is an exponential density. Thus, the Gaussian scale mixture representation of (\ref{eq:lap}) has the Gaussian component density $\pi(x|w)$ in (\ref{eq:gauss_mix_model}) and the mixing density
\begin{equation}\label{eq:exp_mixing_pr}
	\pi(w) \propto \exp\left( -\sum_{i=1}^{d} \lambda_i w_i \right)
\end{equation}
with $w\in\mathbb{R}^d_{>0}$ and rate parameters $\lambda_i = \delta_i^2/2$. 
Since (\ref{eq:lap}) has product-form, we have $\sigpr(w)=\Lambda_{w}$, which denotes a diagonal matrix with $w$ on the main diagonal. 
Note that $\mupr(w)=0$. 
By (\ref{eq:pos_mixing}) we can directly formulate $\pi(w|y)$, where $\pi(w)$ is given by (\ref{eq:exp_mixing_pr}).

\subsection{Dimension-reduced posterior mixing density}
As discussed in Section \ref{sec:sam_pos_mixing}, direct sampling from $\pi(w|y)$ in (\ref{eq:pos_mixing}) is a difficult task.
One option for a dimension-reduced and controlled approximation $\tilde\pi(w|y)$ is to apply the CDR method combined with a transformation of random variables as described in Section \ref{sec:sam_po_comp}.
However since the prior mixing density $\pi(w)$ is exponential and has product-form, it is convenient to apply the CCS method \cite{flockCertifiedCoordinateSelection2024}, a variant of the CDR method.
In the following sections, we first apply the CCS method to find a controlled and easy-to-sample approximation to $\pi(w|y)$.
Then, following the principles of the CCS method, we propose a second dimension-reduced approximation, which is based on a MAP estimate of $\pi(w|y)$.
\subsubsection{CCS-approximated posterior mixing density.}\label{sec:pos_mixCCS}
The CCS method is a special variant of the CDR method as it is developed for the particular case of product-form Laplace or exponential priors.
Here, we employ it to obtain a dimension-reduced approximation $\tilde\pi^\mathrm{CCS}_r(w|y)$ for $\pi(w|y)$. 
Compared to the CDR method, which we briefly introduced in Section \ref{sec:sam_pos_mixing}, in the CCS method the projection matrix $U_r$ is restricted to be a coordinate selection operator, so that $U_r^\mathsf{T} w=w_\mathcal{I}$, where $\mathcal{I}\subseteq\{1,...,d\}$ is the index set of the $r=|\mathcal{I}|$ \emph{selected coordinates}.
Consequently, we denote the set of \emph{not selected coordinates} by $w_\mathcal{J}$ with $\mathcal{J}=\{1,\dots,d\}\backslash \mathcal{I}$. 
We define such a coordinate splitting by a permutation matrix $P\in\mathbb{R}^{d\times d}$:
\begin{equation}\label{eq:coord_spl}
	Pw = \left[\begin{array}{ll}
		w_\mathcal{I}\\ 
		w_\mathcal{J}
	\end{array}\right].
\end{equation}

CCS approximates a posterior density with Laplace or exponential prior by replacing its likelihood with a dimension-reduced ridge approximation on the $r (\ll d)$ selected coordinates. 
This allows the approximation of the posterior mixing density $\pi(w|y)$ in (\ref{eq:pos_mixing}) by
\begin{equation}\label{eq:pos_mix_CCS}
	\tilde\pi_r^{\mathrm{CCS}}(w|y) \propto \tilde\pi_r^{\mathrm{CCS}}(w_\mathcal{I}|y) \pi(w_\mathcal{J}),
\end{equation}
where $\tilde\pi_r^{\mathrm{CCS}}(w_\mathcal{I}|y)$ is an approximate marginal posterior mixing density given by
\begin{equation}\label{eq:pos_mix_CCS_marg}
	\tilde\pi_r^{\mathrm{CCS}}(w_\mathcal{I}|y) \propto \tilde\pi_r(y|w_\mathcal{I}) \pi(w_\mathcal{I})
\end{equation}
with
\begin{equation}\label{eq:red_like_pos_mixing}
	\tilde\pi_r(y|w_\mathcal{I}) =  \left( \int \sqrt{ \pi(y|w_\mathcal{I},w_\mathcal{J}) } \pi(w_\mathcal{J}) \diff{w_\mathcal{J}} \right)^2.
\end{equation}
We denote the posterior approximation which results from replacing $\pi(w|y)$ by $\CCSmix$ by
\begin{equation}\label{eq:pos_CCS}
	\tilde\pi_r^{\mathrm{CCS}}(x|y)=\int \pi(x|w,y) \tilde\pi_r^{\mathrm{CCS}}(w|y) \diff{w}.	
\end{equation}
According to \cite{flockCertifiedCoordinateSelection2024}, the choice of $\tilde\pi_r(y|w_\mathcal{I})$ in (\ref{eq:red_like_pos_mixing}) is optimal with respect to the Hellinger distance.
In particular, we can bound the approximation error as
\begin{equation}\label{eq:bound_ccs}
	\Hell{\pi(w|y)}{\tilde\pi_r^\mathrm{CCS}(w|y)}^2 \leq \epsilon(r) = 2 \sum_{i\in\mathcal{J}} h_i,
\end{equation}
where $h\in\mathbb{R}^d$ is the \emph{diagnostic vector} defined by
\begin{equation}\label{eq:diagnostic}
	h_i= \frac{1}{\lambda_i^2} \int \left( \frac{\partial\log\pi(y|w)}{\partial w_i} \right)^{2} \pi(w|y) \diff{w}.
\end{equation}
Here $\lambda_i$ is the rate parameter of the exponential prior mixing density (\ref{eq:exp_mixing_pr}). 
Recall from (\ref{eq:Hell_pos_mix}) that the control in (\ref{eq:bound_ccs}) ensures 
$$
\Hell{\pi(x|y)}{\tilde\pi_r^{\mathrm{CCS}}(x|y)}^2\leq\epsilon(r).
$$ 

It is clear that the diagnostic $h$ in (\ref{eq:diagnostic}) cannot be computed exactly, since it requires the posterior mixing density to be known. 
Instead, we consider the Monte Carlo estimate
\begin{equation}\label{eq:MC_diagnostic}
	\tilde h_i =\frac{1}{N \lambda_i^2} \sum_{j=1}^N \left( \frac{\partial\log\pi(y|w^{(j)})}{\partial w_i} \right)^{2},
\end{equation}
where $\{w^{(j)}\}_{i=1}^N$ are iid samples from either $\pi(w)$ or some other approximation $\tilde\pi(w|y)$, e.g., $\MAPmix$, see the next section.

Once $h$ is estimated, the coordinate splitting can be performed according to (\ref{eq:bound_ccs}), which suggests that $\mathcal{I}$ should contain the indices associated with the largest components in $h$. 
Hence, for a desired precision $\tau$ on the Hellinger distance, we can attempt to find $\mathcal{I}$ satisfying $\epsilon(r) \leq \tau$. 
However, the resulting number of selected coordinates $r(\tau)$ can be abnormally large, especially if the bound (\ref{eq:bound_ccs}) is loose. 
In this case, we set $r=\min(r(\tau), r_{\max})$ with a pre-given $r_{\max}$ and let $\mathcal{I}$ contain the indices of the $r$ largest components in $h$. 

To allow for the application of standard Markov chain Monte Carlo (MCMC) methods, we suggest the following two reformulations in $\tilde\pi_r^\mathrm{CCS}(w_\mathcal{I}|y)$. 
The resulting expressions of the log target density and its gradient, and an efficient implementation are presented in \ref{proof:CCS_pos_mixing_lap}.

\paragraph{(1) Transformation of random variables.}
The domain of $w$, which is the positive orthant $\mathbb{R}^d_{>0}$, prohibits the application of most standard MCMC methods.
To circumvent this problem, we apply the element-wise transformation of random variables given by $f:\mathbb{R}_{>0} \rightarrow \mathbb{R}$:
\begin{equation}\label{eq:trans_w}
	f(w)=\log(w).
\end{equation}
Then, standard MCMC algorithms can be run on the variable $f(w)$, which now lives in $\mathbb{R}^d$.

\paragraph{(2) Approximation of optimal $\tilde\pi_r(y|w_\mathcal{I})$.}
A practical way to approximate $\tilde\pi_r(y|w_\mathcal{I})$ in (\ref{eq:red_like_pos_mixing}) is to fix the coordinates $w_\mathcal{J}$ at their prior mean as suggested in  \cite{zahmCertifiedDimensionReduction2022b,flockCertifiedCoordinateSelection2024}, i.e.,
\begin{equation*}\label{eq:ap_red_like_pos_mixing}
	w_\mathcal{J}=\left[ \frac{1}{\lambda_{\mathcal{J}_1}}, \cdots, \frac{1}{\lambda_{\mathcal{J}_{d-r}}}\right]^\mathsf{T}.
\end{equation*}
In fact, this approximation allows for further simplifications in $\CCSmix$, which reduce the computations of determinants and inverses of $d \times d$ to $r\times r$ matrices, see \ref{proof:CCS_pos_mixing_lap}.

\vspace{14pt}
We summarize the construction and sampling process of $\CCSmix$ in Algorithm \ref{alg:pos_mix_CCS}.
Note that the main computational effort is reduced to sampling from the low-dimensional approximate marginal posterior mixing density $\CCSmixmarg$ in line \ref{algline:red_pos_mix_1}.
This is because we can sample from $\pi(w_\mathcal{J})$ in line \ref{algline:red_pos_mix_2} independently and in a closed form.
\begin{algorithm}
	Estimate the diagnostic $\tilde{h}$ (\ref{eq:MC_diagnostic}).\;
	Identify the index set $\mathcal{I}$ that contains the indices associated with the $r=\min\{r(\tau), r_{\max}\}$ largest components of $\tilde{h}$.\;
	\For{$i=1$ \KwTo $N$}{
		Sample $w_\mathcal{I}^{(i)} \sim \tilde\pi_r(w_\mathcal{I}|y)$ defined in (\ref{eq:pos_mix_CCS_marg}). \label{algline:red_pos_mix_1}\;
		Sample $w_\mathcal{J}^{(i)} \sim \pi(w_\mathcal{J}) \propto \prod_{i\in\mathcal{J}} \exp \left( -\lambda_i w_i \right)$. \label{algline:red_pos_mix_2}\;
	}
	Re-order coordinates and return: $\{ w^{(i)} \}_{i=1}^N = \left\{ P^\mathsf{T} \left[\begin{array}{l}w_{\mathcal{I}}^{(i)}\\ w_{\mathcal{J}}^{(i)}\end{array}\right]\right\}_{i=1}^N$.
	\caption{Sampling the CCS-approximated posterior mixing density (\ref{eq:pos_mix_CCS})
	}
	\label{alg:pos_mix_CCS}
\end{algorithm}
\subsubsection{MAP-approximated posterior mixing density.}\label{sec:pos_mixMAP}
The following approximation is inspired by the classical Laplace approximation for posterior densities. 
Such an approximation corresponds to a Gaussian whose mean is given by the MAP estimate and whose covariance is given by the inverse of the Hessian of the negative log posterior density evaluated at the MAP estimate \cite{murphyMachineLearningProbabilistic2012}.
In the following approximation to $\pi(w|y)$, we use this concept to construct an approximate marginal posterior density $\MAPmixmarg$ as an alternative to $\CCSmixmarg$ in (\ref{eq:pos_mix_CCS}).
We denote the resulting approximate posterior mixing density by
\begin{equation}\label{eq:pos_mix_MAP}
	\MAPmix \propto \MAPmixmarg \pi(w_\mathcal{J}).
\end{equation} 

We start the construction of $\MAPmixmarg$ by computing a MAP estimate
\begin{equation}\label{eq:w_star}
	w^{\mathrm{MAP}} \in \arg\min_{w\in \mathbb{R}_{\geq0}^d} \, -\log\pi(w|y).
\end{equation}
Extending the domain of $w$ from $\mathbb{R}_{>0}^d$ to $\mathbb{R}_{\geq0}^d$ allows us to perform a coordinate splitting (\ref{eq:coord_spl}) analogous to the previous section based on the following index set:
\begin{equation}\label{eq:I_star}
	\mathcal{I} = \{i\in\{1,\dots,d\} | w^{\mathrm{MAP}}_i>0\}.
\end{equation}
We note that it is not clear if $\pi(w|y)$ is uni-modal, nor if $w^{\mathrm{MAP}}$ is unique.
In our numerical experiments, we always initialize the minimization solver for (\ref{eq:w_star}) at $w_i=0$ for all $i=1,\dots,d$, and the obtained stationary point facilitates good approximations.

After finding $w^{\mathrm{MAP}}$, we define the $r$-dimensional marginal approximate posterior mixing density $\MAPmixmarg$ in (\ref{eq:pos_mix_MAP}) as
\begin{equation}\label{eq:pos_mix_MAP_marg}
	\MAPmixmarg \propto \exp\left( -\frac12 \|w_{\mathcal{I}} - w_{\mathcal{I}}^{\mathrm{MAP}}\|^2_{ {\Sigma^{-1}_{\mathrm{MAP}}} } \right)
	\mathbbm{1} \left( w_{\mathcal{I}} \in \mathbb{R}^r_{>0} \right)
	,
\end{equation}
where $\Sigma_{\mathrm{MAP}}^{-1} =  [-\nabla^2 \log \pi(w^{\mathrm{MAP}}|y)]_{[\mathcal{I}, \mathcal{I}]}$ and $\mathbbm(1)$ is the indicator function.
Here, the subscript $[\mathcal{I}, \mathcal{I}]$ refers to the selection of rows and columns with indices $i\in\mathcal{I}$ such that $\Sigma_{\mathrm{MAP}}\in \mathbb{R}^{r\times r}$, and the formula for $\nabla^2 \log \pi(w|y)$ can be found in \ref{proof:pos_mixing_lap}.
$\MAPmixmarg$ is essentially a truncated Gaussian density on $\mathbb{R}_{>0}^{r}$ centered at $w^{\mathrm{MAP}}_{\mathcal{I}}$.
Sampling from the low-dimensional yet truncated Gaussian density is not straightforward, but special methods, e.g., \cite{botevNormalLawLinear2017a}, are available.
We denote the posterior approximation which results from replacing $\pi(w|y)$ by $\MAPmix$ by
\begin{equation}\label{eq:pos_MAP}
	\tilde\pi_r^{\mathrm{MAP}}(x|y)=\int \pi(x|w,y) \MAPmix \diff{w}.	
\end{equation}

Recall that the MAP of $\pi(w)$ is $w_i=0$ for all $i=1,2,\dots,d$.
Therefore, the coordinate splitting (\ref{eq:I_star}) can be justified by arguing that the coordinates that are most affected by the update from the mixing prior $\pi(w)$ to the mixing posterior $\pi(w|y)$ are those whose MAP estimate changes from $w_i=0$ to $w_i>0$. 
Moreover, we observe in our numerical experiments that if the mixing prior differs significantly from the mixing posterior only on a few coordinates, we have $r \ll d$.
This has two favorable consequences.
First, we can efficiently compute the minimization program (\ref{eq:w_star}), see \ref{proof:pos_mixing_lap_max} for details.
Second, sampling from $\MAPmix$ becomes computationally cheap, since the main sampling effort is reduced to sampling from the truncated $r$-dimensional Gaussian density $\MAPmixmarg$. 
We summarize the construction and sampling from $\MAPmix$ in Algorithm \ref{alg:pos_mix_MAP}.
\begin{algorithm}
	Compute a MAP estimate $w^{\mathrm{MAP}}$ via (\ref{eq:w_star}).\;
	Identify the index set of selected coordinates $\mathcal{I}$ via (\ref{eq:I_star}).\;
	Compute ${\Sigma^{-1}_{\mathrm{MAP}}} =  [-\nabla^2 \log \pi(w^{\mathrm{MAP}}|y)]_{[\mathcal{I}, \mathcal{I}]}$.\;
	\For{$i=1$ \KwTo $N$}{
		Sample $w_{\mathcal{I}}^{(i)} \sim \MAPmixmarg$ defined in (\ref{eq:pos_mix_MAP_marg}). \;
		Sample $w_\mathcal{J}^{(i)} \sim \pi(w_\mathcal{J}) \propto \prod_{i\in\mathcal{J}} \exp \left( -\lambda_i w_i \right)$.\;
	}
	Re-order coordinates and return: $\{ w^{(i)} \}_{i=1}^N = \left\{ P^\mathsf{T} \left[\begin{array}{l}w_{\mathcal{I}}^{(i)}\\ w_{\mathcal{J}}^{(i)}\end{array}\right]\right\}_{i=1}^N$.
	\caption{\\Sampling the MAP-approximated posterior mixing density (\ref{eq:pos_mix_MAP})}
	\label{alg:pos_mix_MAP}
\end{algorithm}
\subsection{Comparison to the Bayesian LASSO}
The Bayesian LASSO \cite{parkBayesianLasso2008d} is designed to handle the LASSO estimate of linear regression parameters in the Bayesian framework. The LASSO posterior density equates to our posterior in (\ref{eq:pos_lap_pr}) if one (i) lets $A$ be a matrix of standardized regressors, (ii) assumes the variance of the noise to be constant across parameters, i.e., $\sigobs=\sigma^2 I$, $\sigma^2>0$, and (iii) fixes the rate parameter globally, i.e., $\delta_i=\delta>0$ for all $i$.

Our sampling method is different from the Bayesian LASSO, since the latter employs a hierarchical framework by modeling the noise variance $\sigma^2$ as hyper-parameter and conditioning the Laplace prior of $x$ on $\sigma^2$. This modeling choice is made to ensure uni-modality of the posterior density. Moreover, the global rate parameter $\delta$ can be fixed in advance via empirical Bayes or also included in the hierarchical framework as hyper-parameter. In our notation, the hierarchical model of the Bayesian LASSO reads
\begin{eqnarray}\label{eq:bay_lasso}
	\eqalign{
		y|x, \sigma^2 &\sim \mathcal{N}(Ax,\sigma^2 I) \\
		x|\sigma^2, w_1, w_2, \dots, w_d &\sim \mathcal{N}(0,\sigma^2 \Lambda_w) \\
		\sigma^2, w_1, w_2, \dots, w_d &\sim \pi(\sigma^2) \diff{\sigma^2} \prod_{i=1}^d \frac{\delta^2}{2} \exp\left(- \frac{\delta^2 w_i^2}{2} \right) \diff{w_i}
	}
\end{eqnarray}

The parameters in the hierarchical model (\ref{eq:bay_lasso}) can be sampled with a Gibbs sampler in closed form by exploiting conjugacy relationships. Sampling the full conditional $x|\sigma^2, \delta, w_1, w_2, \dots, w_d$ is similar to sampling our Gaussian posterior component density but with the slightly different posterior component precision matrix $\sigpo(w)^{-1} = \sigma^{-1} [ A^\mathsf{T} A + \Lambda_{w}^{-1} ]$, which is due to the conditioning of $x$ on $\sigma^2$.

Another difference worth mentioning from a sampling perspective is the way in which the vector of variances $w$ is sampled. While the Bayesian LASSO can sample all full conditionals $w_i|x_i,\sigma^2,\delta^2$ independently, we must sample all $w_i$ from $\pi(w|y)$ at once. Hence, it is challenging to apply our general sampling Algorithm \ref{alg:gen} for the exact posterior in high-dimensional problems, but it also motivates the use of the dimension-reduced approximations $\tilde\pi_r^\mathrm{CCS}(w|y)$ in (\ref{eq:pos_mix_CCS}) or $\tilde\pi_r^\mathrm{MAP}(w|y)$ in (\ref{eq:pos_mix_MAP}). 
A clear advantage of our method is that as soon as $w$ is sampled, sampling $x$ from the posterior component density $\pi(x|w,y)$ can be performed in parallel. 
%
\section{Numerical experiments}\label{sec:num_exp}
In the following numerical experiments, we test our methods from Section \ref{sec:app_lap} on posterior densities with Laplace prior.
Besides showing the feasibility of our approximations (\ref{eq:pos_CCS}) and (\ref{eq:pos_MAP}), we aim to highlight the benefits of approximating $\pi(w|y)$ instead of directly approximating $\pi(x|y)$.

Specifically, in the first experiment we compare two CCS-approximations of $\pi(x|y)$: one obtained via CCS on $\pi(w|y)$, see Section \ref{sec:pos_mixCCS}, and the other one obtained by applying CCS directly on $\pi(x|y)$.
To this end, we write CCS(W) or CCS(X), depending on whether we apply the CCS method to $\pi(w|y)$ or $\pi(x|y)$.
We differentiate the resulting posterior approximations by using the notations
\begin{eqnarray*}
	\CCSposW \propto \int \pi(x|w,y) \CCSmixmarg \pi(w_\mathcal{J}) \diff{w} \\
	\CCSposX \propto \CCSmixmargX \pi(x_\mathcal{J}).
\end{eqnarray*}
Here $\CCSmixmargX$ is obtained analogously to $\CCSmixmarg$ by replacing $w$ by $x$ in (\ref{eq:pos_mix_CCS_marg}) and (\ref{eq:red_like_pos_mixing}).
Moreover, the diagnostic for CCS(X) is computed as in (\ref{eq:diagnostic}) with $w$ replaced by $x$, and $\lambda_i$ replaced by $\delta_i$.

In the second experiment, we compare two MAP-approximations of $\pi(x|y)$: one obtained via a Laplace approximation on the selected coordinates $w_\mathcal{I}$ of $\pi(w|y)$, see Section \ref{sec:pos_mixMAP}, and the other one obtained via a Laplace approximation directly on $\pi(x|y)$.
To differentiate between the resulting posterior approximations, we denote them by 
\begin{eqnarray*}
	\MAPposW \propto \int \pi(x|w,y) \MAPmixmarg \pi(w_\mathcal{J}) \diff{w} \\
	\MAPposX \propto \exp\left(-\frac12 \| x - x^\mathrm{MAP} \|^2_{ \tilde H(x^\mathrm{MAP})}\right),
\end{eqnarray*}
where $x^\mathrm{MAP}\in\mathbb{R}^d$ is the MAP estimate of $\pi(x|y)$, and $H(x^\mathrm{MAP})\in\mathbb{R}^{d\times d}$ is an approximation to $-\nabla^2 \log \pi(x^\mathrm{MAP}|y)$, see Section \ref{sec:pos_mixMAP}.

In some experiments, we apply the Metropolis-adjusted Langevin algorithm (MALA) \cite{robertsExponentialConvergenceLangevin1996e} for sampling.
MALA is derived by discretizing a Langevin diffusion equation and steers the sampling process by using gradient information of the log-target density \cite{robertMonteCarloStatistical2004}.
Whenever we use MALA, we compute 5 independent chains and check their convergence by means of the estimated potential scale reduction (EPSR) statistic \cite{gelmanInferenceIterativeSimulation1992}.
In brief, the EPSR compares the within-variance of the sample chains with their in-between-variance, and chains with an EPSR of less than $1.1$ are considered to have converged. 
We use the Python package arviz \cite{Kumar2019} to compute the ESPR, effective sample sizes (ESS) and credibility intervals (CI) (see, e.g., \cite{murphyMachineLearningProbabilistic2012} for definitions).

\subsection{1D signal deblurring}
In this section, we use a relatively low dimensional test problem to study the behavior of our methods. Due to the low dimensionality, we can compute reference samples $x^{(i)}\sim\pi(x|y)$ and $w^{(i)}\sim \pi(w|y)$ by running MALA on the full-dimensional posterior $\pi(x|y)$ and full-dimensional posterior mixing density $\pi(w|y)$, respectively. 
For both reference sample sets, we compute a total of $25,000$ samples with an ESS of approximately $5,000$. 
Note that we apply the random variable transformation (\ref{eq:trans_w}) to be able to employ MALA for sampling $\pi(w|y)$.
\subsubsection{Problem description.}
The data $y$ is obtained artificially via
\begin{eqnarray*}
	y=G s_\mathrm{true}+e,
\end{eqnarray*}
where $s_{\mathrm{true}} \in \mathbb{R}^{1024}$ denotes the piece-wise constant ground truth, $G$ is a Gaussian blur operator with the kernel width $27$ and standard deviation $3$, and $e \in \mathbb{R}^{1024}$ is a realization from $\mathcal{N}(0, \sigma_\mathrm{obs}^2)$ with $\sigma_\mathrm{obs}=0.03$. The true signal and the data are shown in Figure \ref{fig:1D_ex_problem}.
\begin{figure}[h]%
	\centering
	\begin{minipage}[t]{.48\textwidth}
		\centering
		\includegraphics[]{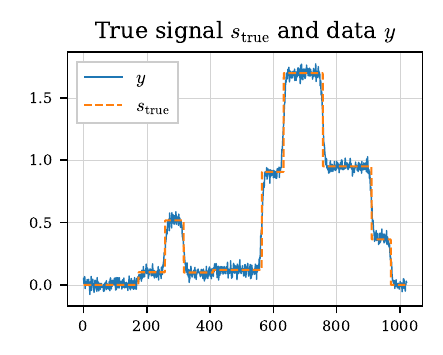}
		\caption{True signal and data of the 1D example.}
		\label{fig:1D_ex_problem}
	\end{minipage}%
	\hfill
	\begin{minipage}[t]{.48\textwidth}
		\centering
		\includegraphics[]{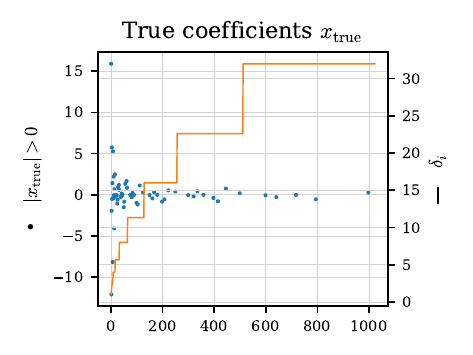}
		\caption{Scatters (left y-axis): Haar wavelet coefficients of the true signal. Solid line (right y-axis): Rate parameters $\delta_i$ given in (\ref{eq:delta}).}
		\label{fig:1D_ex_rate_coeff}
	\end{minipage}
\end{figure}

We employ a $10$-level Haar wavelet transform with periodic boundary condition and formulate the Bayesian inverse problem in the coefficients domain. Let $W$ and $W^\dagger$ denote the discrete wavelet transform and the inverse discrete wavelet transform, respectively. The true coefficients $x_\mathrm{true}=W s_\mathrm{true}$ are sparse with $\|x_\mathrm{true}\|_0=60$, see Figure \ref{fig:1D_ex_rate_coeff}. The posterior density formulated with respect to the coefficients reads
\begin{equation}\label{eq:1D_ex}
	\pi(x|y) \propto \exp\left( -\frac{1}{2 \sigma_\mathrm{obs}^2} \| y - G W^\dagger x \|_2^2 - \sum_{i=1}^{1024} \delta_i |x_i | \right).
\end{equation}
Thus, following our previous notation for the forward operator, we have $A=GW^\dagger$. We use the Python package pywt \cite{Lee2019} to compute the discrete wavelet transforms.

We define the rate parameters $\delta_i$ such that our prior in (\ref{eq:1D_ex}) on $x$ corresponds to a Besov-$\mathcal{B}_{11}^1$ prior \cite{KolSparse, lassasDiscretizationinvariantBayesianInversionand2009} on the signal $s$. In this way, we naturally obtain a vector of rate parameters $\delta_i$, which takes into account the different scales of the wavelet coefficients as
\begin{equation}\label{eq:delta}
	\delta_i=2^{ \frac{\ell(i)}{2}},
\end{equation}
where $\ell(i) \in \{1,\dots,10\}$ denotes the level of the $i$-th wavelet coefficient. We plot $\delta_i$ in Figure \ref{fig:1D_ex_rate_coeff}.
\subsubsection{MAP-approximated posterior mixing density.}\label{sec:1D_approx_sam}
In the following, we use Algorithm \ref{alg:pos_mix_MAP} to compute $5,000$ samples $x^{(i)}\sim\MAPposW$.
We first sample $w^{(i)}\sim\MAPmix$ given in (\ref{eq:pos_mix_MAP}), and then draw $5,000$ samples $x^{(i)}\sim\pi(x|w^{(i)},y)$ using linear RTO (see \ref{proof:lin_RTO}).

In order to set up the MAP-based approximation $\MAPmix$, we first need to compute $w^{\mathrm{MAP}}$ in (\ref{eq:w_star}). 
To this end, we use the L-BFGS-B algorithm \cite{byrd1995limited} with the initial setting $w_i=0$ for all $i=1,\dots,d$. 
We obtain a sparse result with $\|w^{\mathrm{MAP}}\|_0=57$, where non-zero components are mostly concentrated on the lower indices: $\max_i \{i|w^{\mathrm{MAP}}_i \neq 0\} = 125$. 
Due to the sparseness in $w^\mathrm{MAP}$, sampling from $\MAPmix$ is computationally cheap.
In particular, the main computational cost is sampling from a $57$-dimensional truncated multivariate Gaussian.
We achieve this by using the minimax tilting method from \cite{botevNormalLawLinear2017a}, which yields independent samples. 
Consequently, we draw $x^{(i)}\sim\pi(x|w^{(i)},y)$ using linear RTO.

To compare the samples of $\MAPposW$ with the reference samples, we compute their means and bounds of the $60\%$ and $90\%$ CIs in signal space and plot the results in Figure \ref{fig:1D_ex_post} (a).
Visually, we can confirm that our results are very close to the reference solution. Moreover, we note that the samples from $\MAPposW$ are uncorrelated.
\begin{figure}[h]
	\centering
	\includegraphics[]{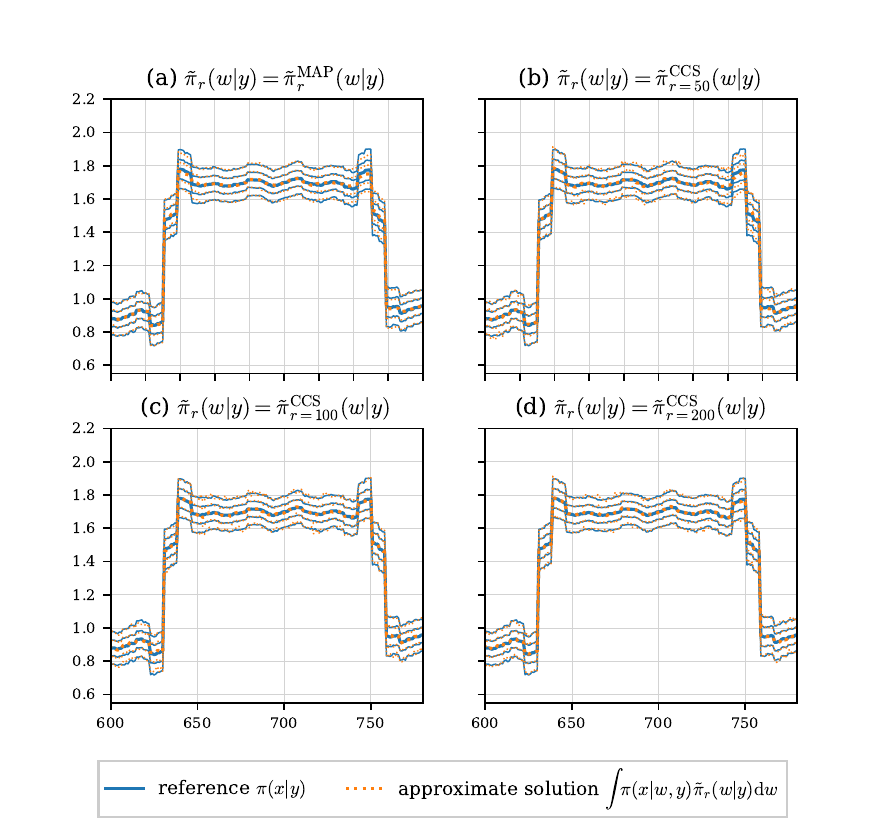}
	\caption{Sample means (thick lines) and bounds of the $60\%$ and $90\%$ CI (thin lines) for the signal $s=W^\dagger x$.
		The sampling results from our method are shown as dotted lines, and the reference solution is displayed as solid line. 
		For (a), we use the samples obtained in Section \ref{sec:1D_approx_sam} and for (b), (c), and (d), the samples from Section \ref{sec:1D_red_sam}. 
		We only show a section of the signal for better visual comparison.} 
	\label{fig:1D_ex_post}
\end{figure}
\subsubsection{CCS-approximated posterior mixing density.}\label{sec:1D_red_sam}

We now compare CCS(W) and CCS(X) in terms of closeness to the exact posterior and sample correlation.
To ensure optimal conditions for both approximations, we compute Monte Carlo estimates $\tilde h_{\mathrm{ref},\mathrm{w}}$ and $\tilde h_{\mathrm{ref},\mathrm{x}}$ of the exact diagnostics using the reference samples $w^{(i)} \sim \pi(w|y)$ and $x^{(i)} \sim \pi(x|y)$.
The resulting bounds in the Hellinger distance (\ref{eq:bound_ccs}) are computed by first sorting the diagnostics in ascending order and then calculating their cumulative sums.
We show the bounds in Figure \ref{fig:1D_ex_bounds_sort}, where we see that we can achieve error bounds on the posterior approximation of several magnitudes smaller with CCS(W) compared to CCS(X).
This is because the bound of CCS(W) in turn bounds the corresponding posterior approximation $\CCSposW$ due to Proposition \ref{prop:bound_pos_mix}.
\begin{figure}[h]
	\centering
	\includegraphics[]{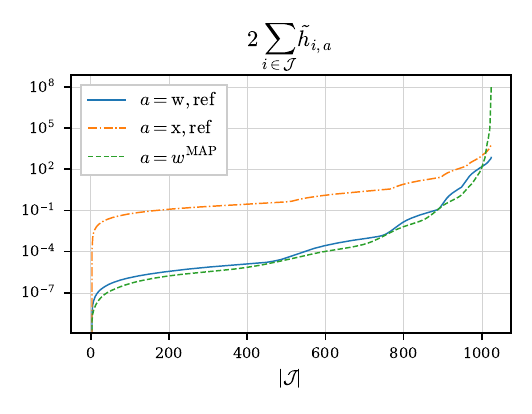}
	\caption{Bound on the Hellinger distance (\ref{eq:bound_ccs}) computed via the reference diagnostics $\tilde{h}_{\mathrm{ref},\mathrm{w}}$ and $\tilde{h}_{\mathrm{ref},\mathrm{x}}$, and the approximation $\tilde{h}_{w^\mathrm{MAP}}$.
		}
	\label{fig:1D_ex_bounds_sort}
\end{figure}
Additionally, we compute $\tilde h_{w^\mathrm{MAP}}$ using the samples $w^{(i)} \sim \tilde\pi_r^\mathrm{MAP}(w|y)$ obtained in the previous section.
The corresponding bound is also shown in Figure \ref{fig:1D_ex_bounds_sort} and yields a fairly good approximation to the reference $\tilde h_{\mathrm{ref}, \mathrm{w}}$.

We now perform the different coordinate splittings $r\in\{50,100,200\}$ for CCS(W) and CCS(X). 
The index sets $\mathcal{I}$ for CCS(W) and CCS(X) are performed according to $\tilde{h}_{\mathrm{ref},\mathrm{w}}$ and $\tilde{h}_{\mathrm{ref},\mathrm{x}}$, respectively. 
Then, we sample the dimension-reduced densities $\CCSmixmarg$ and $\CCSmixmargX$, where we approximate the optimal reduced likelihood by setting the not selected coordinates to the respective prior mean in both densities.
That is, we set $x_{\mathcal{J},i}=0$ in CCS(X) and $w_{\mathcal{J},i}=1/\lambda_i$ for $i=1,\dots,d-r$ in CCS(W).
For CCS(W), we use Algorithm \ref{alg:pos_mix_CCS} to compute $5,000$ samples $w^{(i)}\sim\tilde\pi_r^\mathrm{CCS}(w|y)$, and then use linear RTO to draw $5,000$ samples $x^{(i)}\sim\pi(x|w^{(i)},y)$.
For CCS(X), we directly use MALA to compute $5,000$ samples $x^{(i)}\sim\CCSposX$.

To compare the performance of CCS(W) with CCS(X) in terms of sample correlation, we compute the mean normalized ESS ($\mathrm{nESS}=\mathrm{ESS}/N$, for $N$ samples) for the sample sets shown in the rows of Table \ref{tab:1D_ESS}.
One of the key findings here is that the two-step sampling structure of Algorithm \ref{alg:gen} can significantly decrease the correlation in the posterior samples $\{x^{(i)}\}_{i\geq1}$.
This can be seen when comparing the mean nESS of the sample sets $\{w_\mathcal{I}^{(i)}\}_{i\geq1}$ (first row) and $\{x_\mathcal{I}^{(i)}\}_{i\geq1}$ (second row).
Moreover, when comparing the mean nESS of the sample set $\{x_\mathcal{I}^{(i)}\}_{i\geq1}$ obtained by CCS(W) (second row) and CCS(X) (third row), we observe that CCS(W) yields significantly less correlated samples in the selected coordinates $x_\mathcal{I}$.
%
\begin{table}[ht]
	\caption{Mean nESS for different sample sets obtained via CCS(W) and CCS(X) with number of selected coordinates $r\in\{50,100,200\}$.
		We note that the standard deviation over the sample chains is very small for all sample sets.
	}
	\label{tab:1D_ESS}
	\begin{center}
		\begin{tabular}{@{}lllll}
			\br
			$r$& &$50$ &$100$ &$200$ \\
			
			\mr
			\multirow{2}{*}{CCS(W)}  
			&$w_\mathcal{I}^{(i)}\sim\tilde\pi_r^\mathrm{CCS}(w_\mathcal{I}|y)$ &0.36 &0.24 &0.16 \\
			&$x_\mathcal{I}^{(i)}\sim\pi(x_\mathcal{I}|w^{(i)},y)$ &0.84 &0.78 &0.78\\
			
			
			\mr
			CCS(X)
			&$x_\mathcal{I}^{(i)}\sim\tilde\pi_r^\mathrm{CCS(X)}(x_\mathcal{I}|y)$ &0.70 &0.23 &0.10\\
			
			\br
		\end{tabular}
	\end{center}
\end{table}

Finally, we visually compare our sampling results from CCS(W) with the reference solution in Figure \ref{fig:1D_ex_post}, see plots (b), (c), and (d).
It is remarkable that even for $r=50$, we are already very close to the reference solution, and for $r=100$ we practically recover the reference solution, such that selecting more coordinates has only little benefits on the result.

\subsection{2D super-resolution microscopy}
In this experiment, we illustrate the feasibility of our approximated posterior mixing density $\MAPmix$ in a high-dimensional problem. 
The test problem is inspired by the application of stochastic optical reconstruction microscopy (STORM) in \cite{zhuFasterSTORMUsing2012}.
A similar example was considered in the Bayesian context in \cite{durmusEfficientBayesianComputation2018}. 
STORM is a super-resolution microscopy technique based on single-molecule stochastic switching, where the goal is to detect molecule positions in live cell imaging. 
The images are obtained by a microscope detecting the photon count of the (fluorescence) photoactivated molecules.
\subsubsection{Problem description.}
We consider a microscopic image $y\in\mathbb{R}^m$, which is obtained from a 2D pixel-array by concatenation in the usual column-wise fashion. 
Here, we set $m=32^2=1,024$. 
In STORM, we want to estimate precise molecule positions by computing a super-resolution image $x\in\mathbb{R}^{d}$. 
In this example, we set the oversampling ratio $k=4$, which leads to $d=mk^2=16,384$. 
Based on the kernel from the optical measurement instrument given in \cite{zhuFasterSTORMUsing2012}, we generate the forward operator $A\in\mathbb{R}^{m\times d}$. 
The data $y$ is obtained via
\begin{equation}\label{eq:storm_data}
	y= A x_\mathrm{true} + e,
\end{equation}
where $e \in \mathbb{R}^m$ is simulated from $\mathcal{N}(0, \sigma_\mathrm{obs}^2)$.

Similar as in \cite{zhuFasterSTORMUsing2012}, we generate the ground-truth image $x_\mathrm{true}$ for the high photon count case with $50$ uniformly distributed molecules on a field of the size $4\mu\mathrm{m}\times 4\mu\mathrm{m}$. 
The intensity of each molecule is simulated from a log-normal distribution with the mode $3,000$ and the standard deviation $1,700$. 
In Figure \ref{fig:storm_data_truth_map}, we show the ground-truth image and the data, which is obtained according to (\ref{eq:storm_data}) with $\sigma_\mathrm{obs}=30$. 
We use a Laplace prior due to the sparsity of $x_\mathrm{true}$, which leads to the posterior density
\begin{equation}
	\pi(x) \propto \exp\left( -\frac{1}{2 \sigma_\mathrm{obs}^2 } \| y - A x \|_2^2 - \delta \|x\|_1 \right).
\end{equation}
We chose $\delta=1.275$ based on the visual quality of different MAP estimate computations.
\begin{figure*}[ht]
	\centering
	\includegraphics[]{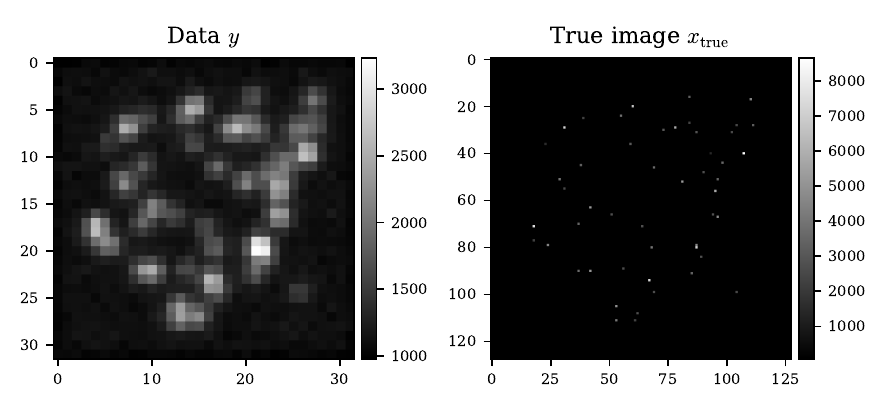}
	\caption{Left: Data $y$ computed via (\ref{eq:storm_data}). Right: True image $x_\mathrm{true}$.}
	\label{fig:storm_data_truth_map}
\end{figure*}
%

\subsubsection{MAP-based approximation of the posterior mixing density.}\label{sec:1D_approx_sam_2}
In the following, we compare our posterior approximation $\MAPposW$ in (\ref{eq:pos_MAP}) with a Laplace approximation of $\pi(x|y)$, which we denote by $\MAPposX$. 
Recall that we write MAP(W) and MAP(X) to differentiate between the two posterior approximation methods.

We first compute approximate posterior samples via MAP(X).
We start by constructing the Laplace approximation $\MAPposX$ by computing the MAP estimate of $\pi(x|y)$.
This is achieved by solving $x^\mathrm{MAP} = \argmin_x - \log\pi(x|y)$, which is a convex optimization problem, and we use the convex optimization toolbox cvxpy \cite{diamond2016cvxpy} for the computation.
The MAP estimate $x^\mathrm{MAP}$ is shown in Figure \ref{fig:Faster_STORM_MAP}, and we observe that it recovers the true molecule positions quite well.
\begin{figure}[ht]
	\centering
	\begin{minipage}[t]{0.47\textwidth}
		\vspace{0pt}
		\centering
		\includegraphics[]{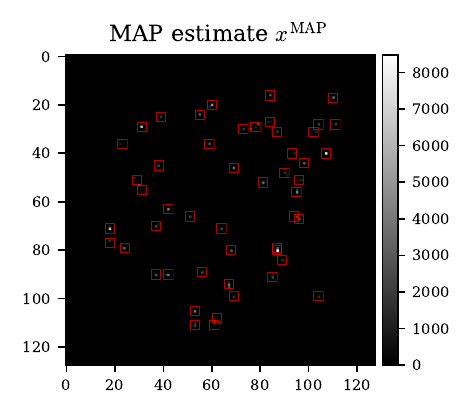}
		\caption{MAP estimate $x^\mathrm{MAP}$. The true molecule positions are located in the centers of the red boxes.}
		\label{fig:Faster_STORM_MAP}
	\end{minipage}\hfill
	\begin{minipage}[t]{0.47\textwidth}
		\vspace{0pt}
		\centering
		\includegraphics[]{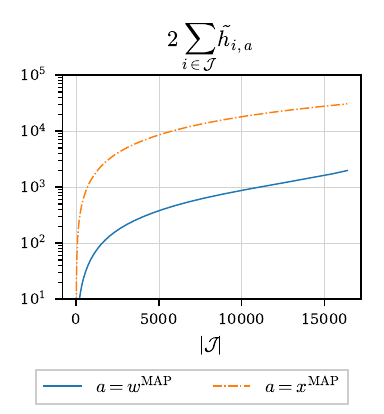}
		\caption{Bound on the Hellinger distance (\ref{eq:bound_ccs}) for CCS(W) and CCS(X), estimated via the approximate diagnostics $\tilde{h}_{w^\mathrm{MAP}}$ and $\tilde{h}_{x^\mathrm{MAP}}$, respectively.
			}
		\label{fig:Faster_STORM_diagnostics}
	\end{minipage}
\end{figure}

We cannot directly proceed with the construction of $\MAPposX$, since the Laplace approximation requires the log posterior density to be two times differentiable at the MAP estimate.
This is not the case for the present posterior density, and we approximate the Laplace prior as
$\pi(x) \propto \exp (-\delta \|x\|_1) \approx \exp(-\delta \sum_{i=1}^d \sqrt{x_i^2 + \gamma})$ for some $0 < \gamma \ll 1$.
We follow the heuristic rule in \cite{flockCertifiedCoordinateSelection2024} and fix $\gamma=\delta^2/4$.
The requirement of this approximation is a disadvantage compared to MAP(W), where we employ the Laplace approximation on the selected coordinates $w_\mathcal{I}$ of the smooth density $\pi(w|y)$.
With the modified prior density, the approximated Hessian of the negative log posterior density at $x^\mathrm{MAP}$ reads 
$\tilde H(x^\mathrm{MAP}) = 1/\sigma_\mathrm{obs}^2 A^\mathsf{T} A + \Lambda_z$,
where $\Lambda_z\in\mathbb{R}^{d\times d}$ is a diagonal matrix with the vector 
$z=\delta \gamma / \sqrt{(x^\mathrm{MAP})^{2} + \gamma}\in\mathbb{R}^{d}$ on its diagonal.
Thus, $\MAPposX=\mathcal{N}(x^\mathrm{MAP}, \tilde H(x^\mathrm{MAP})^{-1})$, and we draw $10,000$ samples $x^{(i)}\sim\MAPposX$ using linear RTO.

We now compute approximate posterior samples via MAP(W).
Recall that sampling from $\MAPposW$ consists of sampling $w^{(i)} \sim \tilde\pi_r^{\mathrm{MAP}}(w|y)$ and then sampling $x^{(i)}\sim\pi(x|w^{(i)},y)$.
To obtain $\MAPmix$, we first compute $w^\mathrm{MAP}$ via the formulations in \ref{proof:pos_mixing_lap_max}. 
As in the experiment before, we use the L-BFGS-B algorithm \cite{byrd1995limited} with the initial guess $w_i=0$ for all $i=1,\dots,d$.
We obtain $\|w^\ast\|_0=138$, which is very sparse considering the total dimension of $d=16,384$. 
Using Algorithm \ref{alg:pos_mix_MAP}, we then compute $10,000$ samples of $w^{(i)} \sim \tilde\pi_r^{\mathrm{MAP}}(w|y)$, where we use the minimax tilting method from \cite{botevNormalLawLinear2017a} to sample from the $138$-dimensional truncated multivariate Gaussian $\MAPmixmarg$.
Subsequently, we draw $10,000$ samples $x^{(i)}\sim\pi(x|w^{(i)},y)$ via linear RTO.

With the samples $w^{(i)} \sim \tilde\pi_r^{\mathrm{MAP}}(w|y)$ and $x^{(i)}\sim\MAPposX$, we compute approximate diagnostics $\tilde{h}_{w^\mathrm{MAP}}$ and $\tilde{h}_{x^\mathrm{MAP}}$. 
In Figure \ref{fig:Faster_STORM_diagnostics}, we compare the corresponding upper bounds on the Hellinger distance for CCS(W) and CCS(X).
Even though the bounds are abnormally large, the bound corresponding to $\tilde{h}_{w^\mathrm{MAP}}$ is orders of magnitude smaller than the one corresponding to $\tilde{h}_{x^\mathrm{MAP}}$.
This indicates that CCS(W) may yield a better posterior approximation than CCS(X), since the bound of CCS(W) in turn bounds the resulting posterior approximation $\CCSposW$ due to Proposition \ref{prop:bound_pos_mix}.

In Figure \ref{fig:storm_mean_ci}, we compare the sample means and differences between the bounds of the 99\% sample CIs of the samples $x^{(i)}\sim\MAPposW$ and $x^{(i)}\sim\MAPposX$.
We observe that the true molecule positions are identified equally well by both posterior means.
Comparing the sample CIs, we note that the CI widths of $\MAPposW$ are tighter at the molecule positions but larger everywhere else.
\begin{figure*}[ht]%
	\centering
	\includegraphics[]{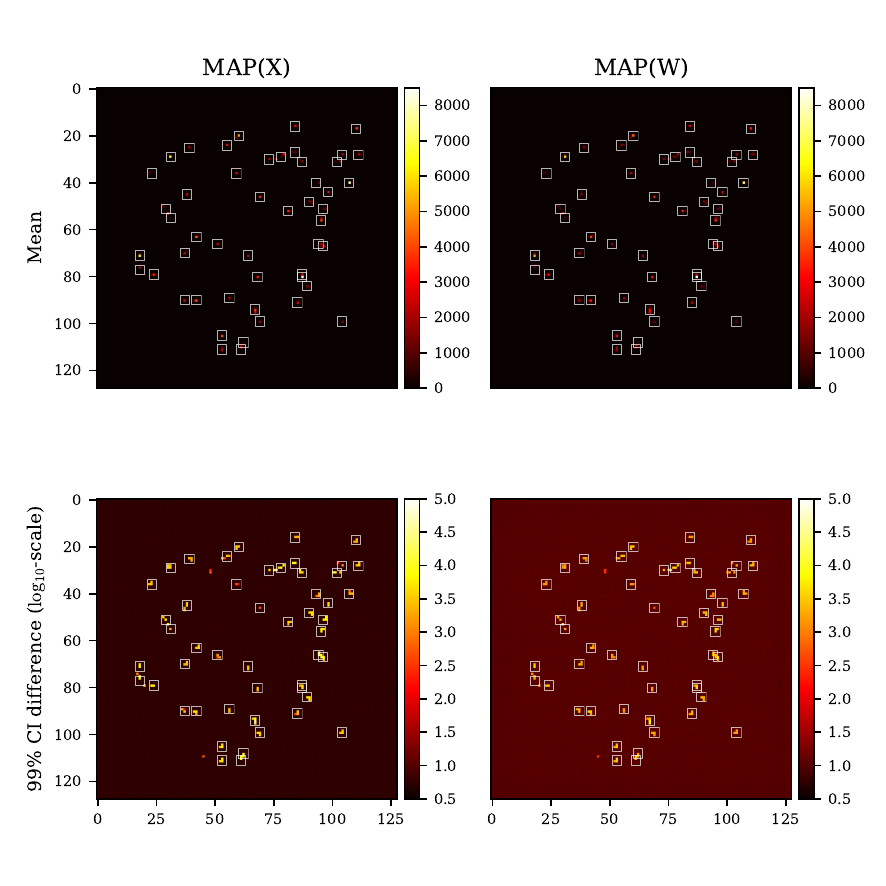}
	\caption{Left column: Results from sampling $\MAPposX$. 
		Right column: Results from sampling $\MAPposW$. 
		First row: Sample mean. 
		Second row: Difference between the bounds of the 99\% sample CI. 
		The true molecule positions are located in the centers of the white boxes.}
	\label{fig:storm_mean_ci}
\end{figure*}

To obtain more insights on the difference between the approximations, we plot three exemplary sections at the image coordinates $x_\mathrm{image}=18$, $x_\mathrm{image}=87$, and $y_\mathrm{image}=28$ in Figure \ref{fig:Faster_STORM_mean_CI_sect}.
We observe that the mean of $\MAPposX$, i.e., $x^\mathrm{MAP}$, is close to the truth, which however, does not mean that it is close to the true posterior mean.
Moreover, since $\MAPposX$ is Gaussian, its CI bounds are symmetrical around the mean and may reach wide into unrealistic negative photon counts, e.g., see the sections at $x_\mathrm{image}=18$ and $y_\mathrm{image}=28$.
On the contrary, the solution given by $\MAPposW$ is more realistic as most of the probability mass is concentrated on the positive photon counts.
Furthermore, the fact that the sample mean of $\MAPposW$ mostly underestimates the true photon numbers and that they are not necessarily contained in the 99\% sample CI can be explained by the typical shrinkage property of the Laplace prior.
\begin{figure*}[ht]%
	\centering
	\includegraphics[]{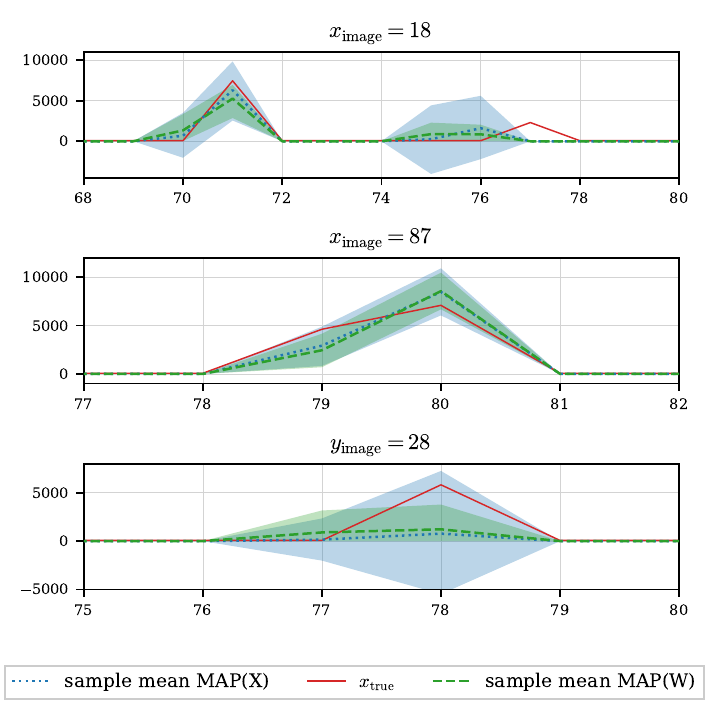}
	\caption{Means and differences between the bounds of the 99\% CI estimated via samples from $\MAPposW$ and $\MAPposX$ at the image coordinates $x_\mathrm{image}=18$, $x_\mathrm{image}=87$, and $y_\mathrm{image}=28$.}
	\label{fig:Faster_STORM_mean_CI_sect}
\end{figure*}

\section{Conclusions}\label{sec:conc}
We derived the closed form expressions of the posterior component density and posterior mixing density in the case of a linear-Gaussian likelihood and a Gaussian mixture prior. 
Our proposed sampling strategy for such a posterior includes two steps. First, sampling the mixing variables from the posterior mixing density, and second, sampling the component densities conditioned on the sampled mixing variables. 
Since direct sampling from the posterior mixing density is often not possible, alternative methods like MCMC methods must be used.
However, exact sampling via such methods is relatively expensive due to the computational complexity of the mixing density. 
Therefore, we propose to replace the posterior mixing density by a dimension-reduced approximation and provide a bound in the Hellinger distance for the resulting posterior approximation.

Furthermore, we considered the Laplace prior, which can be written as a continuous Gaussian mixture, and described two methods using the CCS method to approximate the posterior mixing density. 
In our numerical results, we showed that approximating the posterior mixing density via these methods, instead of directly approximating the posterior density, yields samples closer to the exact posterior.
Moreover, due to the two-step sampling scheme, the correlation among posterior samples can be significantly reduced.

As a future research direction, we aim to investigate whether the numerical improvements observed on problems with the Laplace prior can also be achieved with other continuous Gaussian mixture priors, such as the Student's-t or horseshoe prior.
\section{Acknowledgments}
RF and YD were supported by a Villum Investigator grant (no. 25893) from The Villum Foundation. FU was supported by the Research Council of Finland through the Flagship of Advanced Mathematics for Sensing, Imaging and Modelling, and the Centre of Excellence of Inverse Modelling and Imaging (decision numbers 359183 and 353095, respectively).
\appendix
\setcounter{section}{0}
\section{Proofs}
\subsection{\texorpdfstring{Proposition \ref{prop:pos_mix}}{Proposition 1}}\label{proof:pos_mix}
Under Assumption \ref{ass:gen_pos_mix_ass}, we have the conditional independence
\begin{equation}\label{eq:cond_ind_proofs}
	\pi(y,w|x) = \pi(y|x)\pi(w|x)
\end{equation}
and can directly obtain
\begin{eqnarray*}
	\pi(w|y)
	&= \frac{\pi(w,y)}{\pi(y)} \\
	&= \frac{\int \pi(w,y|x)\pi(x)\diff{ x}}{\pi(y)} \qquad \mathrm{with}\ \mathrm{(\ref{eq:cond_ind_proofs})} \\
	&= \frac{\int \pi(w|x)\pi(y|x)\pi(x)\diff{x}}{\pi(y)} \\
	&\propto \pi(w)\int \pi(y|x)\pi(x|w)\diff{ x},
\end{eqnarray*}
which is the expression given in (\ref{eq:gen_pos_mix_2}). 
Moreover,
\begin{eqnarray*}
	\pi(x|y,w)
	&= \frac{\pi(y,w|x)\pi(x)}{\pi(y,w)} \qquad \mathrm{with}\ \mathrm{(\ref{eq:cond_ind_proofs})} \\
	&= \frac{\pi(y|x)\pi(w|x)\pi(x)}{\pi(y,w)} \\
	&\propto \pi(y|x)\pi(x|w),
\end{eqnarray*}
which is the expression given in (\ref{eq:gen_pos_mix_1}). 
\subsection{\texorpdfstring{Proposition \ref{prop:bound_pos_mix}}{Proposition 2}}\label{proof:hell_bound}
Consider
\begin{eqnarray*}
	\pi(x|y) &= \int \pi(x|w,y)\pi(w|y)\diff{w} \\
	\tilde\pi(x|y) &= \int \pi(x|w,y)\tilde\pi(w|y)\diff{w}.
\end{eqnarray*}
Suppose that we have control of the Hellinger distance such that
\begin{equation*}
	\Hell{\pi(w|y)}{\tilde\pi(w|y)}^2 =  \frac{1}{2} \int( \sqrt{\pi(w|y)}-\sqrt{\tilde\pi(w|y)} )^2\diff{w} \leq \epsilon.
\end{equation*}

Then, a Cauchy--Schwarz inequality permits to write
\begin{eqnarray*}\fl
	\Hell{\pi(x|y)}{\tilde\pi(x|y))}^2 &= 1 - \int \sqrt{\pi(x|y)} \sqrt{\tilde\pi(x|y)} \diff{x} \\
	&= 1 - \int \sqrt{ \int \pi(x|w,y)\pi(w|y) \diff{w} } \sqrt{ \int \pi(x|w,y) \tilde\pi(w|y)\diff{w} } \diff{x} \\
	&{\leq} 1 - \int \int \sqrt{ \pi(x|w,y)\pi(w|y) } \sqrt{ \pi(x|w,y)\tilde\pi(w|y) } \diff{w} \diff{x} \\
	&= 1 - \int \int \pi(x|w,y) \diff{x} \sqrt{ \pi(w|y) \tilde\pi(w|y) } \diff{w}  \\
	&= 1 - \int \sqrt{ \pi(w|y) \tilde\pi(w|y) } \diff{w}  \leq \varepsilon.
\end{eqnarray*}
\subsection{\texorpdfstring{Proposition \ref{prop:Gauss_pos_mix}}{Proposition 3}}\label{proof:Gauss_pos_mix}
%
According to Proposition \ref{prop:pos_mix} we have the posterior mixture
\begin{equation*}
	\pi(x|y) = \int \pi(x|w,y) \pi(x|w) \diff{w},
\end{equation*}
where
\begin{eqnarray*}
	\pi(x|w,y) &\propto \pi(y|x)\pi(x|w),\\
	\pi(w|y) &\propto \pi(w)\int \pi(y|x')\pi(x'|w) \diff{x'}. 
\end{eqnarray*}
We know that $\pi(y|x)$ is a linear-Gaussian likelihood and $\pi(x|w)$ is a Gaussian prior conditioned on $w$. They are defined as
\begin{equation*}
	\pi(y|x)=\frac{1}{\sqrt{(2\pi)^m\det\sigobs}} \exp\left( -\frac12 \|Ax-y\|_{\sigobs^{-1}}^2 \right),
\end{equation*}
and
\begin{equation*}
	\pi(x|w)=\frac{1}{\sqrt{(2\pi)^d\det\sigpr(w)}} \exp\left( -\frac12 \| x -\mupr(w) \|_{\sigpr(w)^{-1}}^2 \right).
\end{equation*}
Then, $\pi(x|w,y)$ is also Gaussian, and the precision matrix and mean are given as
\begin{eqnarray}
	\sigpo(w)^{-1} &= A^\mathsf{T} \sigobs^{-1}A + \sigpr(w)^{-1} \label{eq_app:conpo_sig} \\
	\mupo(w,y) &= \sigpo(w)\left( A^\mathsf{T} \sigobs^{-1}y + \sigpr(w)^{-1} \mupr(w) \right). \label{eq_app:conpo_mean}
\end{eqnarray}

We now derive the expression for $\pi(w|y)$. 
Using the explicit densities of $\pi(y|x)$ and $\pi(x|w)$, we have 
\begin{eqnarray*}\fl\eqalign{
		\pi(w|y)
		&= \frac{\pi(w)\int \pi(y|x)\pi(x|w) \diff{ x } }{\pi(y)} \\
		&\propto \frac{\pi(w)}{\sqrt{\det\sigpr(w)}}\int \exp\left( -\frac12 \|Ax-y\|_{\sigobs^{-1}}^2  -\frac12 \| x  -\mupr(w) \|_{\sigpr(w)^{-1}}^2 \right)
		\diff{ x  }.
	}
\end{eqnarray*}
Employing (\ref{eq_app:conpo_sig}) and (\ref{eq_app:conpo_mean}), we can write
\begin{eqnarray*}\fl\eqalign{
		&-\frac12 \|Ax-y\|_{\sigobs^{-1}}^2  -\frac12 \| x  -\mupr(w) \|_{\sigpr(w)^{-1}}^2  \\
		&= -\frac12 x^\mathsf{T} \big(A^\mathsf{T}\sigobs^{-1}A + \sigpr(w)^{-1}\big) x + x^\mathsf{T}  A^\mathsf{T}\sigobs^{-1}y -\frac12\|y\|_{\sigobs^{-1}}^2  - \frac12 \|\mupr(w)\|_{\sigpr(w)^{-1}}^2 + x^\mathsf{T} \sigpr(w)^{-1} \mupr(w) \\
		&= -\frac12 \| x - \mupo(w,y)\|_{\sigpo(w)^{-1}}^2 +\frac12\|\mupo(w,y)\|_{\sigpo(w)^{-1}}^2- \frac12\|y\|_{\sigobs^{-1}}^2 - \frac12 \|\mupr(w)\|_{\sigpr(w)^{-1}}^2.
	}
\end{eqnarray*}
Hence, we derive
\begin{eqnarray*}\fl\eqalign{
		\pi(w|y)
		&\propto \frac{\pi(w)}{\sqrt{\det\sigpr(w)}}\int \exp\left( -\frac12 \|Ax-y\|_{\sigobs^{-1}}^2  -\frac12 \| x -\mupr(w)  \|_{\sigpr(w)^{-1}}^2 \right) \diff{ x } \\
		&\propto \frac{\pi(w)}{\sqrt{\det\sigpr(w)}} \exp\left(\frac12\|\mupo(w,y)\|_{\sigpo(w)^{-1}}^2\right) \exp\left(- \frac12 \|\mupr(w)\|_{\sigpr(w)^{-1}}^2\right)  \\
		& \qquad\cdot \int \exp\left( -\frac12 \| x - \mupo(w,y)\|_{\sigpo(w)^{-1}}^2 \right)  \diff{ x }  \\
		&\propto \frac{\pi(w)}{\sqrt{\det\sigpr(w)}} \exp\left(\frac12\|\mupo(w,y)\|_{\sigpo(w)^{-1}}^2\right) \exp\left(- \frac12 \|\mupr(w)\|_{\sigpr(w)^{-1}}^2\right)  \sqrt{\det \sigpo(w)}. \\
	}
\end{eqnarray*}
\section{Linear RTO}\label{proof:lin_RTO}
Here we describe linear RTO for sampling from $\mathcal{N}(\mupo(w^{(i)},y), \sigpo(w^{(i)}))$. 
In the following we drop the superscripts ${(i)}$ for clarity. 

By the definition of $\sigpo(w)$ in (\ref{eq:conpo_sig}) and $\mupo(w,y)$ in (\ref{eq:conpo_mean}), we have that $x$, given by
\begin{eqnarray} \label{eq:RTO_pos_comp} \fl
	x = \mupo(w,y) + \sigpo(w)^{1/2}\xi = \sigpo(w)\left( A^\mathsf{T} \sigobs^{-1}y + \sigpr(w)^{-1} \mupr(w) + \sigpo(w)^{-1/2}\xi \right),
\end{eqnarray}
satisfies $x\sim\mathcal{N}(\mupo(w,y), \sigpo(w))$ whenever $\xi\sim\mathcal{N}(0,I_d)$. 
In order to compute (\ref{eq:RTO_pos_comp}) without evaluating the potentially expensive factorization $\sigpo(w)^{1/2}$ (or $\sigpo(w)^{-1/2}$), linear RTO proceeds as follows. 

Consider the decompositions $\sigobs^{-1} = L_1 L_1^\mathsf{T}$ and $\sigpr(w)^{-1} = L_2(w) L_2(w)^\mathsf{T}$. 
Note that $L_1$ can be pre-computed and recycled for any new sample $w$. 
Also, when $\sigpr(w)$ is a diagonal matrix (as in our application in Section \ref{sec:app_lap}), the computation of $L_2(w)$ is trivial. 
Thus, the random vector $u=A^\mathsf{T} L_1 \zeta + L_2(w) \gamma$, where $\zeta \sim \mathcal{N}(0,I_m) $ and $ \gamma \sim \mathcal{N}(0,I_d) $ satisfies $u\sim  \sigpo(w)^{-1/2}\xi $, so that (\ref{eq:RTO_pos_comp}) can be rewritten as
\begin{equation}\label{eq:x_lRTO}
	\sigpo(w)^{-1} x =  A^\mathsf{T} \sigobs^{-1}y + \sigpr(w)^{-1} \mupr(w) + A^\mathsf{T} L_1 \zeta + L_2(w) \gamma.
\end{equation}
%
Finally, we let $M(w)^\mathsf{T} = \left[A^\mathsf{T} L_1, L_2(w)\right]$ and $z(w)^\mathsf{T} = \left[y^\mathsf{T} L_1 + \zeta^\mathsf{T}, \mupr(w)^\mathsf{T} L_2(w)  + \gamma^\mathsf{T} \right]$, such that (\ref{eq:x_lRTO}) corresponds to the optimality condition for the following linear least squares problem:
\begin{equation}\label{eq:CGLS}
	x = \argmin_{a\in \mathbb{R}^{d}} \norm{ M(w)a - z(w) }_2^2.
\end{equation}
\section{Explicit expressions for the case of Laplace prior}

\subsection{The gradient and the Hessian of the posterior mixing density}\label{proof:pos_mixing_lap}
Based on (\ref{eq:pos_mixing}) and (\ref{eq:exp_mixing_pr}) we obtain the log posterior mixing density for the product-form Laplace prior as
\begin{eqnarray}\label{eq:log_mix}\fl\eqalign{
		\log \pi(w|y) =
		&- \sum_{i=1}^{d} \lambda_i w_i  - \frac12 \log \det \left[\Lambda_w^{\frac12} A^\mathsf{T} \sigobs^{-1} A \Lambda_w^{\frac12} + I_d\right] \\
		&+ \frac12 y^\mathsf{T} \sigobs^{-1} A \Lambda_w^{\frac12}\left[\Lambda_w^{\frac12}A^\mathsf{T} \sigobs^{-1} A\Lambda_w^{\frac12} + I_d\right]^{-1}\Lambda_w^{\frac12} A^\mathsf{T} \sigobs^{-1} y - \log c_{W|Y=y},
	}
\end{eqnarray}
where $I_d$ denotes the $d$-by-$d$ identity matrix, $\Lambda_w$ denotes a diagonal matrix with $w$ on the diagonal, and $c_{W|Y=y}$ is a normalizing constant such that $\pi(w|y)$ is a probability density. 

The gradient reads
\begin{equation}\label{eq:grad_log_mix}
	\nabla \log \pi(w|y)
	= -\lambda  - \frac12 \diag{\left[\left(A^\mathsf{T} \sigobs^{-1} A\right)^{-1}+\Lambda_w\right]^{-1}} + \frac12 z^{\odot2},
\end{equation}
and the Hessian is
\begin{eqnarray*}\fl\eqalign{
		\nabla^2 \log \pi(w|y) = 
		\frac12 \left( \left[\left(A^\mathsf{T} \sigobs^{-1} A\right)^{-1}+\Lambda_w\right]^{-1} \right)^{\odot 2} 
		- \Lambda_z \left[\left(A^\mathsf{T} \sigobs^{-1} A\right)^{-1}+\Lambda_w\right]^{-1} \Lambda_z,
	}
\end{eqnarray*}
where $z= \left[\left(A^\mathsf{T} \sigobs^{-1} A\right)^{-1}+\Lambda_w\right]^{-1} \left[A^\mathsf{T} \sigobs^{-1} A\right]^{-1} A^\mathsf{T} \sigobs^{-1} y$.
Moreover, $\diag{B}$ draws the main diagonal from the matrix $B$ as a column vector, and $\cdot^{\odot 2}$ denotes the element-wise square.
\subsection{MAP of the posterior mixing density}\label{proof:pos_mixing_lap_max}
We solve
\begin{equation}\label{eq:MAP}
	\min_{w\in[0,\infty)^d} -\log \pi(w|y)
\end{equation}
in order to obtain a MAP estimate of the posterior mixing density. 
We set the initial $w$ as the $0$-vector and apply L-BFGS-B introduced in \cite{byrd1995limited} to solve (\ref{eq:MAP}). 

The main computational cost during the optimization lies in the calculations of the determinant and the inverse of $\Lambda_w^{\frac12}A^\mathsf{T} \sigobs^{-1} A \Lambda_w^{\frac12} + I_d$ in $\log \pi(w|y)$ and the inverse of $\Lambda_w + \left(A^\mathsf{T} \sigobs^{-1} A\right)^{-1}$ in $\nabla\log\pi(w|y)$ at each iteration. However, in every iteration we expect that the most of $w_i$ for $i=1, \cdots, d$ remain at the initial value, i.e., 0, so that we can reduce the computational complexity by using the sparsity of $w$. 
The details are as follows. 
Note that we do not use an iteration index for simplicity.

Set $\mathcal{I} = \{i|w_i>0\}$ to include all $r=|\mathcal{I}|$ indices that correspond to the informed coordinates in $w$, where $r\ll d$. 
In addition, we define $\mathcal{J}=\{1, \cdots, d\}\backslash \mathcal{I}$ to contain all non-informed indices, and set $w_j=0$ for $j\in\mathcal{J}$. 
Based on $\mathcal{I}$, we define a permutation matrix $P$ such that the first $r$ elements of the vector $Pw$ are the coordinates $w_i$ with $i\in\mathcal{I}$. 

Then, we use the permutation to efficiently compute the log posterior mixing density (\ref{eq:log_mix}) as follows.
\[
P\left(\Lambda_w^{\frac12}A^\mathsf{T} \sigobs^{-1} A \Lambda_w^{\frac12} + I_d\right)P^\mathsf{T} = \left[\begin{array}{ll} B & 0 \\ 0 & I_{d-r}\end{array}\right],
\]
where $0$ denotes the matrices with all elements equal to 0, and $B=P\Lambda_{w_{\mathcal{I}}}^{\frac12}A^\mathsf{T} \sigobs^{-1} A \Lambda_{w_{\mathcal{I}}}^{\frac12}P^\mathsf{T}\in\mathbb{R}^{r\times r}$. By using the properties of the permutation matrix, we get
\begin{eqnarray*}
	\det\left[\Lambda_w^{\frac12} A^\mathsf{T} \sigobs^{-1} A \Lambda_w^{\frac12} + I_d\right] & = \det\left[P\left(\Lambda_w^{\frac12}A^\mathsf{T} \sigobs^{-1} A \Lambda_w^{\frac12} + I_d\right)P^\mathsf{T}\right]
	= \det(B),\\
	\left(\Lambda_w^{\frac12}A^\mathsf{T} \sigobs^{-1} A \Lambda_w^{\frac12} + I_d\right)^{-1} & = 
	P^\mathsf{T}\left[\begin{array}{ll} B^{-1} & 0 \\ 0 & I_{d-r}\end{array}\right]P.
\end{eqnarray*}
Due to $r\ll d$, we note that the size of $B$ is much smaller than $\Lambda_w^{\frac12}A^\mathsf{T} \sigobs^{-1} A \Lambda_w^{\frac12} + I_d$. 
Hence, it is much cheaper to compute its determinant and its inverse. 

To compute the inverse of $\left(A^\mathsf{T} \sigobs^{-1} A\right)^{-1}+\Lambda_w$ for the gradient of the log posterior mixing density (\ref{eq:grad_log_mix}), we apply the Woodbury matrix identity \cite{higham2002accuracy}. 
We define $\hat A = A^\mathsf{T} \sigobs^{-1} A$, and partition the permutation matrix $P$ into 
\[
P=\left[\begin{array}{l}P_{\mathcal{I}}\\ P_{\mathcal{J}}\end{array}\right]
\]
with $P_{\mathcal{I}}\in\mathbb{R}^{r\times d}$ and $P_{\mathcal{J}}\in\mathbb{R}^{(d-r)\times d}$. Then, we can rewrite
\[
\left(A^\mathsf{T} \sigobs^{-1} A\right)^{-1}+\Lambda_w = \hat{A}^{-1}+P_{\mathcal{I}}^\mathsf{T}\Lambda_{w_{\mathcal{I}}}P_{\mathcal{I}}.
\]
Applying the Woodbury formula, we obtain
\begin{equation}
	\left[\hat{A}^{-1}+P_{\mathcal{I}}^\mathsf{T}\Lambda_{w_{\mathcal{I}}}P_{\mathcal{I}}\right]^{-1} = \hat{A} - \hat{A}P_{\mathcal{I}}^\mathsf{T}\left(\Lambda_{w_{\mathcal{I}}}^{-1} +P_{\mathcal{I}}\hat{A}P_{\mathcal{I}}^\mathsf{T}\right)^{-1}P_{\mathcal{I}}\hat{A}.
\end{equation}
Note that $\Lambda_{w_{\mathcal{I}}}^{-1} +P_{\mathcal{I}}\hat{A}P_{\mathcal{I}}^\mathsf{T}$ is a $r$-by-$r$ matrix with $r \ll d$, whose inverse is computationally much cheaper than the inverse of $\hat{A}^{-1}+\Lambda_w$.
\subsection{CCS-approximated posterior mixing density}\label{proof:CCS_pos_mixing_lap}
In order to sample the CCS-approximated posterior mixing density, we first use the logarithmic function to transform $w$, i.e.,
\begin{equation*}
	v_i =\log(w_i).
\end{equation*}
Then, the log posterior mixing density for the product-form Laplace prior reads
\begin{eqnarray*}\fl
	\log \pi(v|y)
	=
	&-\sum_{i=1}^{d} \lambda_i \exp(v_i) - \frac12 \log \det \left[ A^\mathsf{T} \sigobs^{-1} A + \Lambda_{\exp(v)}^{-1} \right] \\
	&+ \frac12 y^\mathsf{T} \sigobs^{-1} A \left[ A^\mathsf{T} \sigobs^{-1} A + \Lambda_{\exp(v)}^{-1} \right]^{-1} A^\mathsf{T} \sigobs^{-1} y + \frac12 \sum_{i=1}^d v_i - \log c_{V|Y=y},
\end{eqnarray*}
where $c_{V|Y=y}$ is a normalizing constant such that $\pi(v|y)$ is a probability density. Its gradient is
\begin{eqnarray*}\fl
	\nabla \log \pi(v|y) =
	&- \Lambda_{\lambda} \exp(v) - \frac12 \diag{ \left[A^\mathsf{T} \sigobs^{-1} A + \Lambda_{\exp(v)}^{-1}  \right]^{-1} A^\mathsf{T} \sigobs^{-1} A } \\
	&+ \frac12 \left( \Lambda_{\exp(v)}^{-\frac12} \left[ A^\mathsf{T} \sigobs^{-1} A + \Lambda_{\exp(v)}^{-1}  \right]^{-1} A^\mathsf{T} \sigobs^{-1} y \right)^{\odot 2}  + 1.
\end{eqnarray*}

Now we approximate the optimal $\tilde \pi(y|v_\mathcal{I})$ in (\ref{eq:red_like_pos_mixing}) by fixing $v_\mathcal{J}$ in $\tilde \pi(y|v)$ at their prior mean, i.e., we fix 
\begin{equation}\label{eqn:v_j}
	v_\mathcal{J} = \left[ -\log \lambda_{\mathcal{J}_1}, \dots, -\log \lambda_{\mathcal{J}_{d-r}} \right]^\mathsf{T}.
\end{equation}
Obviously, the main computational cost is associated to the calculations of the determinant and the inverse of $A^\mathsf{T} \sigobs^{-1} A + \Lambda_{\exp(v)}^{-1}$. Since $v_\mathcal{J}$ is fixed, only $r$ elements on the main diagonal in $A^\mathsf{T} \sigobs^{-1} A + \Lambda_{\exp(v)}^{-1}$ are changing during the sampling. Therefore, we can re-order its elements and keep the computations mainly on a $r$-by-$r$ submatrix for each sample. We give the details as follows.

Similar as in \ref{proof:pos_mixing_lap_max}, we define the permutation matrix $P$ for a coordinate splitting (\ref{eq:coord_spl}) based on an estimate of the diagnostic vector $h$ (\ref{eq:MC_diagnostic}). Thus, the first $r$ elements of the vector $Pv$ are the coordinates $v_i$ with $i\in\mathcal{I}$, and the remaining coordinates are the coordinates $v_j$ with $j\in\mathcal{J}$. By using $P$, we obtain
$$
P \left[ A^\mathsf{T} \sigobs^{-1} A + \Lambda_{\exp(v)}^{-1} \right] P^\mathsf{T} = \left[
\begin{array}{ll} 
	B + \Lambda_{\exp(v_\mathcal{I})}^{-1} & C \\ 
	C^\mathsf{T} & D
\end{array} \right],
$$
with the constant matrices $B=[A^\mathsf{T} \sigobs^{-1} A]_{\mathcal{I},\mathcal{I}}$, $C=[A^\mathsf{T} \sigobs^{-1} A]_{\mathcal{I},\mathcal{J}}$, and $D=[A^\mathsf{T} \sigobs^{-1} A]_{\mathcal{J},\mathcal{J}}+ \Lambda_{\lambda_\mathcal{J}}$.
Now we define 
$$
Z_{v_\mathcal{I}}=\Lambda_{\exp(v_\mathcal{I})}^{-1}+B - C D^{-1} C^\mathsf{T}\in\mathbb{R}^{r\times r},
$$
where $B - C D^{-1} C^\mathsf{T}$ is constant and can be pre-computed. 
Based on the Schur complement, we know
\begin{eqnarray*}\fl\eqalign{
		\det \left[ A^\mathsf{T} \sigobs^{-1} A + \Lambda_{\exp(v)}^{-1} \right] =& \det(D)\det\left(Z_{v_\mathcal{I}}\right),\\
		\left[ A^\mathsf{T} \sigobs^{-1} A + \Lambda_{\exp(v)}^{-1}  \right]^{-1} =& P^\mathsf{T} \left[\begin{array}{ll}
			Z_{v_\mathcal{I}}^{-1} & - Z_{v_\mathcal{I}}^{-1}CD^{-1}\\
			-D^{-1}C^TZ_{v_\mathcal{I}}^{-1} & 
			D^{-1}+D^{-1}C^TZ_{v_\mathcal{I}}^{-1}CD^{-1}\end{array}\right] P.}
\end{eqnarray*}
This way, we only need to compute the determinant and the inverse of the $r$-by-$r$ matrix $Z_{v_\mathcal{I}}$ for each new $v_\mathcal{I}$.
%
%
\printbibliography
\end{document}